\documentclass[10pt,journal, letterpaper]{IEEEtran}

\IEEEoverridecommandlockouts
\usepackage{array}
\usepackage{booktabs}
\usepackage{graphicx}
\usepackage{mathrsfs}
\usepackage{amsfonts}
\usepackage{amssymb}
\usepackage{amsmath}
\usepackage{graphicx}
\usepackage{multirow}
\usepackage{amsthm}
\usepackage{color}
\usepackage[table]{xcolor}
\usepackage{arydshln}
\usepackage[lined,vlined,ruled,commentsnumbered]{algorithm2e}
\usepackage{epstopdf}
\usepackage{pdfsync}
\usepackage{stmaryrd}
\usepackage{url}
\usepackage{cite}
\usepackage{bm}
\usepackage{indentfirst}
\usepackage{epsfig}
\usepackage{xcolor}
		
\newtheorem{theorem}{Theorem}
\newtheorem{lemma}{Lemma}

\newtheorem{definition}{Definition}
\newtheorem{proposition}{Proposition}

\newtheorem{assumption}{Assumption}
\newtheorem{mechanism}{Mechanism}

\newtheorem{corollary}{Corollary}
\newtheorem{framework}{Framework}
\newtheorem{example}{Example}
\ifodd 0
\newcommand{\rev}[1]{{\color{black}#1}}
\newcommand{\revh}[1]{{\color{black}#1}}
\newcommand{\revi}[1]{{\color{black}#1}}
\newcommand{\revr}[1]{{\color{black}#1}}
\newcommand{\revk}[1]{{\color{black}#1}}
\newcommand{\revg}[1]{{\color{blue}#1}}
\newcommand{\com}[1]{\textbf{\color{red} (Comment: #1) }}
\newcommand{\comg}[1]{\textbf{\color{blue} (COMMENT: #1)}}
\newcommand{\response}[1]{\textbf{\color{blue} (RESPONSE: #1)}}
\else
\newcommand{\rev}[1]{{#1}}
\newcommand{\revh}[1]{{#1}}
\newcommand{\revi}[1]{{#1}}
\newcommand{\revr}[1]{{#1}}
\newcommand{\revk}[1]{{#1}}
\newcommand{\revg}[1]{{#1}}
\newcommand{\com}[1]{}
\newcommand{\comg}[1]{}
\newcommand{\response}[1]{}
\fi

\def\eq{\triangleq}

\def\N{\mathcal{N}}                 
\def\M{\mathcal{M}}
\def\R{\mathcal{R}}                  
\def\T{\mathcal{T}}                  

\def\tseg{\beta}                       
\def\bufsize{B}                         

\def\downloader{n}
\def\receiver{m}
\def\br{r}
\def\brvector{\boldsymbol{r}}
\def\segamount{\kappa}
\def\starttime{t}
\def\endtime{\tau}

\def\Cost{C}

\def\Utility{U}
\def\SW{W}

\def\va{\theta}                          
\def\curbuf{B^{\textsc{cur}}}     
\def\prebr{R^{\textsc{pre}}}       
\def\Valueq{V^{\textsc{q}}}       
\def\valueq{v^{\textsc{q}}}
\def\Valueb{V^{\textsc{b}}}        
\def\Lossdeg{L^{\textsc{qd}}}	 
\def\lossdeg{l^{\textsc{qd}}}	
\def\llossdeg{\tilde{l}^{\textsc{qd}}}	

\def\auctioneer{n}
\def\bidder{m}

\def\objamount{K}
\def\obj{k}

\def\bid{\boldsymbol{b}}           
\def\brmatrix{\boldsymbol{R}}
\def\pr{p}                                   
\def\prvector{{\boldsymbol{p}}}

\def\Win{\Gamma}
\def\Pay{\Pi}

\def\payoff{P}

\def\win{\sigma^{\dag}}
\def\winset{\boldsymbol{\sigma}^{\dag} }
\def\wpay{\pi^{\dag}} 
\def\wpayset{\boldsymbol{\pi}^{\dag}} 
\def\wbr{\br^{\dag}}
\def\wbrset{\boldsymbol{\br}^{\dag}}
\def\wamount{\widetilde{\kappa}}
\def\winsetupdate{\boldsymbol{\sigma}^{\ddag}}
\def\wbrsetupdate{\boldsymbol{r}^{\ddag}}
\def\wbrupdate{\widetilde{\boldsymbol{r}}}

\def\score{\phi}
\def\mscore{S}
\def\mscoreset{\boldsymbol{S}}
\def\mscoreex{\hat{S}}
\def\mscoresetex{\hat{\boldsymbol{S}}}
\def\mscorei{S^{\dag}}
\def\mscoreseti{\boldsymbol{S}^{\dag}}

\begin{document}

\title{Multi-Dimensional Auction Mechanisms for Crowdsourced Mobile Video Streaming}

\author{Ming Tang, Haitian Pang, Shou Wang, Lin Gao, Jianwei Huang, Lifeng Sun
		\IEEEcompsocitemizethanks{
			\IEEEcompsocthanksitem
		
			M. Tang and J. Huang {(co-corresponding author)} are with The Chinese University of Hong Kong, Hong Kong,
			E-mail: tm014@ie.cuhk.edu.hk, jwhuang@ie.cuhk.edu.hk;
			L. Gao is with Harbin Institute of Technology (Shenzhen), China,
			E-mail: gaol@hit.edu.cn;
			H. Pang, S. Wang, and L. Sun {(co-corresponding author)} are
			with Tsinghua University, China, E-mail: pht14@mails.tsinghua.edu.cn, wangshousoso@126.com, sunlf@tsinghua.edu.cn.
			
			Part of this work has been presented at IEEE WiOpt \cite{WiOpt16} and IEEE INFOCOM \cite{INFOCOM17}. This work is supported by the General Research Funds (Project Number CUHK 14206315 and CUHK 14219016) established under the University Grant Committee of the Hong Kong Special Administrative Region, China, and the NSFC (Grant Number 61472204 and 61521002).					
		}}

\maketitle

\maketitle
\begin{abstract}
	Crowdsourced mobile video streaming enables nearby mobile video users to aggregate network resources to improve their video streaming performances. 
	{However, users are often selfish and may not be willing to cooperate without proper incentives.} Designing an incentive mechanism for such a scenario is challenging due to the users' asynchronous downloading behaviors and their private valuations for multi-bitrate encoded videos. 
	In this work, we propose both single-object and multi-object multi-dimensional auction mechanisms, through which 
    	users sell the opportunities for downloading single and multiple video segments with multiple bitrates, respectively. Both auction mechanisms can achieve truthfulness (i.e, truthful private information revelation) and efficiency (i.e., social welfare maximization).
	Simulations with real traces show
	that crowdsourced mobile streaming facilitated by the auction mechanisms outperforms noncooperative
	streaming by $48.6\%$ (on average) in terms of social welfare.	To evaluate the real-world performance, we also construct a demo system for crowdsourced mobile streaming and implement our proposed auction mechanism. 
	{Experiments over the demo show that 
		those users who provide resources to others and those users who receive helps can increase their welfares by $15.5\%$ and $35.4\%$ (on average) via cooperation, respectively.} 
\end{abstract}
 
\IEEEpeerreviewmaketitle

\section{Introduction}\label{sec:intro}

\subsection{Background and Motivation}\label{subsec:intro-background}

Mobile video traffic accounted for around $55\%$ of global mobile traffic in 2015, and is expected to grow at an annual rate of $62\%$ between 2015 and 2020 
\cite{CiscoReport}. The increasing video demand requires proper resource allocation  methods to achieve desirable user's quality of experience (QoE) in increasingly congested wireless networks with limited radio resources. 
A key challenge to achieve this is that 
different users can have very different QoE requirements (e.g., depending on device screen sizes and user preferences) and channel conditions (e.g., \rev{3G cellular, 4G cellular, or WiFi links}). 
To resolve and exploit the heterogeneity among users and deal with the potential mismatch of video requirements and channel conditions at the individual user level, researchers have recently proposed  a  \emph{crowdsourced mobile streaming} \rev{(CMS) system} \cite{ABR-multi-Lin2016} that enables mobile users to form cooperative groups and share their network resources for more effective mobile video streaming. 

\rev{The CMS system} is very suitable for the adaptive bitrate video streaming (ABR) technology \cite{abr}, 
a widely used video streaming technology in HTTP networks. In ABR, a video is partitioned into multiple video segments, and each video segment is encoded at multiple bitrates. 
{A video user can choose the bitrate of each segment based on his preference and the real-time network condition. 
	Hence, ABR-based video streaming provides a good amount of flexibility for cooperative downloading in \rev{the CMS system}. 
	
	Figure \ref{fig:model} shows an example of \rev{the CMS system} with three users, where each user watches a unique video hosted by the corresponding server. User C does not have a cellular connection to the Internet, so both user A and user B download user C's segments and forward to user C. User A also downloads a segment for user B, as he has a  better downlink channel (4G) than user B (3G). \revk{In the CMS system, the downloading links  from the Internet to users can be either cellular links or WiFi links, and the connections among users can be either WiFi Direct links or Bluetooth. In order to make the presentation easy to follow, we will refer to all the downloading links as ``cellular links" and all the connections among users as ``WiFi Direct links"; these are merely terminology choices that do not limit the applicability of the system. }

	\begin{figure}[t]
		\centering
		\includegraphics[height=3cm]{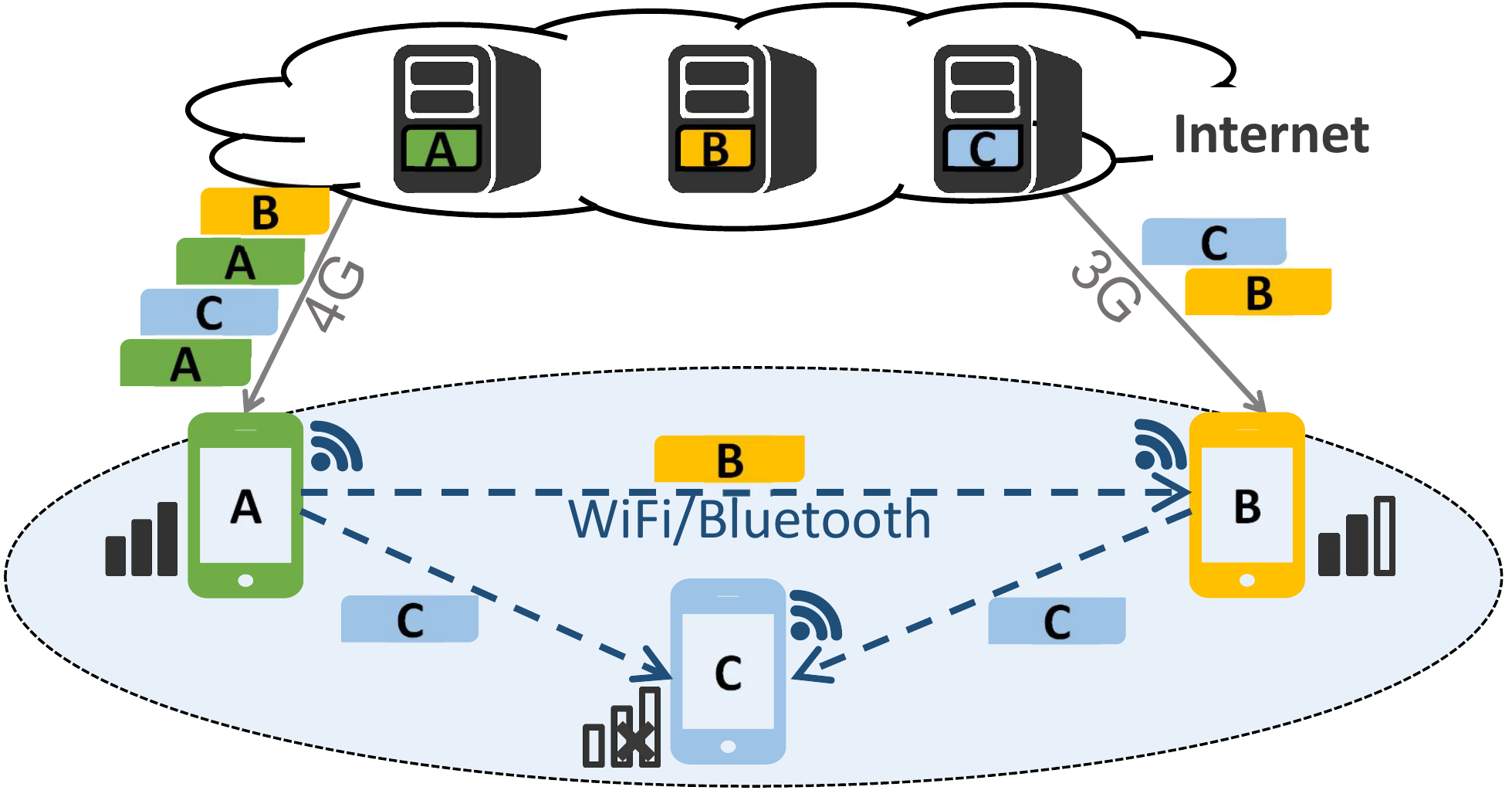}
		{\caption{{Crowdsourced Mobile Streaming.}}\label{fig:model}}
	\end{figure}

	The CMS system is different from the device-to-device (D2D) based \cite{ABR-multi-Keller2012,ABR-multi-Li2015,ABR-multi-Cao14} and peer-to-peer (P2P) based \cite{ABR-P2P-Xu2013, ABR-P2P-Abdallah15,ABR-multi-Klusch14,P2P-incentive} video streaming models, where users \rev{share their \emph{downloaded} video segments with other users through D2D links and the Internet, respectively}. In the CMS system, users share their cellular network resources (for segments downloading), 
	hence it is mainly targeted at the much more common application scenario that different users watch \emph{different} videos}. Different from the bandwidth aggregation (BA) models that {aggregate multiple users' bandwidth to serve one user's streaming need} \cite{ABR-multi-Zhong16,ABR-multi-Seenivasan2014,ABR-multi-Zhang2014}, \rev{the CMS system} aggregates multiple users' bandwidth to satisfy all users' video streaming needs, enhancing the users' QoE through proper network resource allocation.
	
	

	A major challenge for realizing the CMS system is that helping others will increase mobile users' cost, so the mobile users may not be willing to cooperate unless they receive proper incentives. In other words, the success of such a \rev{CMS system} requires a proper \emph{incentive mechanism} that motivates mobile users to crowdsource their network resources for cooperative video segments downloading.

	\subsection{Solution Approach and Contribution}\label{subsec:intro-solution}
	
	In this work, we focus on the \emph{incentive mechanism design} for \rev{the CMS system}.
	Namely, we aim to design such mechanisms that offer enough compensation for each video user to download video segments for others, considering the user's own service request and downloading cost.
	The proposed mechanism needs to consider the following questions for each segment that each user (downloader) downloads:
	\begin{itemize}
		\item \emph{Receiver Selection:}
		Whose segment will the downloader  download?
		
		\item \emph{Bitrate Adaption:}
		What bitrate (quality) will the receiver choose for the segment to be downloaded?
		
		\item \emph{Cost Compensation:}
		How much will the downloader be compensated for his downloading cost by the receiver?
	\end{itemize}
	
	\rev{It is challenging to design an effective incentive mechanism that addresses above questions, because of the users' private valuations for multi-bitrate encoded video segments as well as their asynchronous downloading behaviors}.
	{First, a user's valuation for a segment at a particular bitrate is the user's private information and can vary  over time. The diverse and varying private valuation induces difficulties in evaluating users' contributions in cooperation and determining the proper incentive levels. }
	{Second, video scheduling in ABR is segment based instead of time-slot based, so it is challenging to schedule the downloading cooperation among the users who request and download videos at different times. }
	
	{\rev{Auction is widely used for allocating objects among the users who have private valuations. Hence, we propose  auction-based incentive mechanisms for the CMS  to handle the users' private valuation revelation. 
			To address the asynchronous operations,} we consider decentralized mechanisms: when a user (downloader) is ready to download new segments, he will initiate an auction to decide for whom to download at what bitrate with what price. In other words, {the downloader acts as an auctioneer}, and his nearby users (who request videos) act as bidders, bidding for the segment downloading opportunities.} 
	
	{Classical single-dimensional auction, where a bidder submits a single value indicating his willingness-to-pay, is not applicable in our crowdsourced model.}
	{This is because the video segments are encoded at multiple bitrates in ABR, so a bidder needs to specify multi-dimensional information in the bid, i.e., his intended bitrate and the price he is willing to pay for such a bitrate. This motivates us to consider a multi-dimensional auction  in this work.}
	\begin{figure}
		\center
		\includegraphics[height=2.5cm]{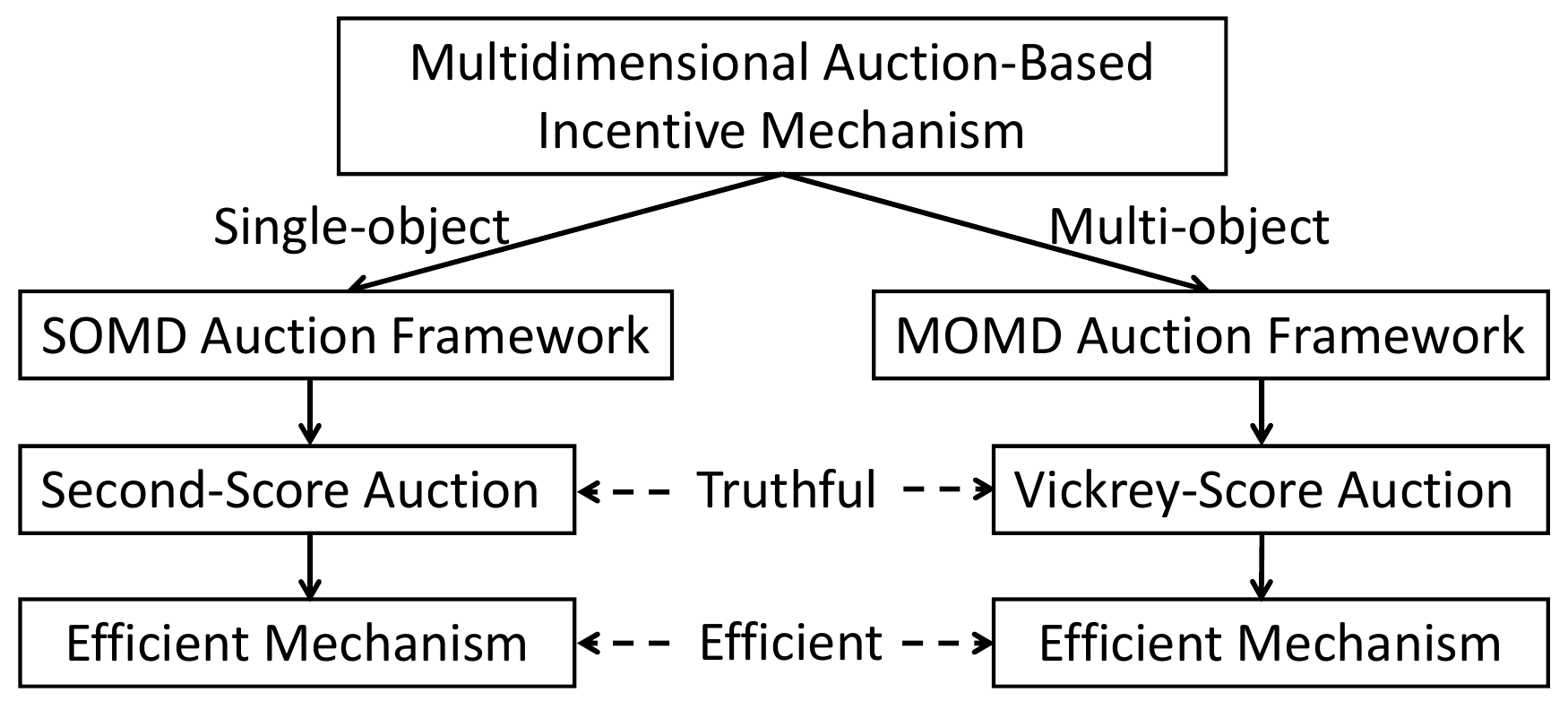}
		\caption{{Theoretical Framework of This Paper.}}\label{fig:logic}
	\end{figure}
	
	As a benchmark, we first propose a \emph{single-object multi-dimensional auction} (SOMD)\cite{Yeon-KooChe1993} based incentive mechanism framework for \rev{the CMS system}, where an auctioneer allocates one segment in one auction. Based on the SOMD framework, we propose a second-score auction-based mechanism that ensures the truthful user valuation revelation in \rev{the CMS system}. Through a proper design of the score function (to be discussed in Section \ref{subsec:single-auction-2ndscore}), we derive the efficient mechanism that maximizes the social welfare. 
	
	However, such a single-object allocation may induce extensive signaling  overhead {because of the frequently initiated auctions
		, which may negatively affect the video streaming performance}.
	This motivates us to consider a \emph{multi-object multi-dimensional auction} that enables auctioneers to allocate multiple segments in one auction. Such a multi-object allocation introduces an additional dimension in the bidding process---the quantity (the number of the segments that a bidder desires), which is preferential dependent\footnote{Dimension $x$ is preferentially dependent of dimension $y$ if the preference of $x$ depends on the preference of $y$ \cite{Olson1995}.} of price. 
	It has been shown in \cite{MartinBichler2000Multi} that designing a multi-dimensional auction with preferential dependent dimensions is extremely difficult, but it turns out to be the problem that we need to solve. 
	In this work, we propose a \emph{multi-object multi-dimensional  } (MOMD) auction framework, which enables bidders to bid for multiple objects (i.e., segments) with different bitrates in each auction. {Within the MOMD framework, we design the allocation rule and payment rule, which leads to a truthful Vickrey-score auction. By a proper design of the score function, we propose an efficient mechanism that maximizes the social welfare. 
		Figure \ref{fig:logic} illustrates the theoretical framework of this paper.}
	
	\rev{The single-object and multi-object mechanisms assume that every user who is close to a downloader (and watches a video) will participate in the auction (when the downloader is ready to sell his downloading opportunities).  Although the mechanisms maximize the social welfare in each auction (under a properly chosen score function), the long-term social welfare across multiple rounds of auction may not necessarily be maximized in some cases. For example, if a downloader's channel condition is very poor (at the time he initiates the auction), then it might be wise for the nearby users to refrain from bidding and wait for a different downloader (with a better channel condition) to become available. 
		Therefore, we will further modify the proposed mechanisms, allowing users to refrain from bidding  according to certain rules, which can improve the overall long-term system performance. }
	
	Our key contributions are summarized as follows:
	\begin{table*}[t]
		\begin{center}
			\caption{{Multi-User Models in Adaptive Bitrate Streaming} } \label{table:reference}
			\begin{tabular}{|c|c|c|c|c|c|c|c|}
				\hline
				\multirow{2}{*}{Reference} & \multirow{2}{*}{Framework Type}&
				\multicolumn{2}{c|}{Model} &
				\multicolumn{3}{c|}{Method} & \multirow{2}{*}{Demo}\\
				\cline{3-7}
				& & Multi-Server & Multi-Video & Multi-Seg per Allocation & Bitrate Adaptation & Incentive &\\
				\hline
				
				
				\cite{ABR-multi-Keller2012} & Device-to-Device& $\surd$ & $\times$ & $\surd$ & $\times$ & $\times$  & $\surd$\\
				\cite{ABR-multi-Li2015,ABR-multi-Cao14} & Device-to-Device& $\times$ & $\times$ & $\times$ & $\times$ & $\surd$  & $\times$\\
				\cite{ABR-P2P-Xu2013,ABR-multi-Klusch14} &Peer-to-Peer& $\surd$ & $\times$ & $\surd$ & $\surd$ & $\times$  & $\surd$ \\
				\cite{ABR-P2P-Abdallah15}&Peer-to-Peer&  $\surd$ & $\times$ & $\surd$ & $\times$ & $\surd$ & $\times$\\
				\cite{ABR-multi-Zhong16,ABR-multi-Zhang2014,ABR-multi-Seenivasan2014} &Bandwidth Aggregation  & $\times$ & $\times$ & $\surd$ & $\times$ & $\surd$ & $\surd$\\
				\cite{ABR-multi-Lin2016}&Crowdsourced&$\surd$ & $\surd$ & $\times$ & $\surd$ & $\times$ & $\times$\\
				\textbf{This Paper}&\textbf{Crowdsourced }   & {$\boldsymbol{\surd}$} & {$\boldsymbol{\surd}$} & {$\boldsymbol{\surd}$} & {$\boldsymbol{\surd}$} & {$\boldsymbol{\surd}$} & {$\boldsymbol{\surd}$}\\
				\hline			
			\end{tabular}
		\end{center}
	\end{table*}
	\begin{itemize}
		\item  \emph{Auction-Based Incentive Mechanisms in \rev{the CMS system}:} We propose multi-dimensional auction based incentive mechanisms for \rev{the CMS system}, supporting the asynchronous downloading and bitrate adapting of video users. 
		The design of such mechanisms is challenging, as it needs to ensure that users  truthfully report multi-dimensional preferentially dependent information.

		\item \emph{Truthful and Efficient Auction}: For single-segment allocation, we propose a SOMD framework, based on which we propose a truthful and efficient mechanism that maximizes the social welfare in each auction. For multi-segment allocation, we propose an MOMD framework, based on which we propose the first mechanism achieving both truthfulness and efficiency in a multi-object multi-dimensional auction. 

		\item \emph{Modified Mechanism}: To \rev{enhance the long-term social welfare of video streaming services}, we further improve the proposed mechanisms by allowing bidders to \rev{refrain from bidding} according to their current situations. The simulation results show that such modification can successfully decrease rebuffer and bitrate degradation frequency along the entire video streaming.
		
		\item \emph{Real-world Demonstration System}: We construct a demo system using Raspberry PI 
		(a series of single-board computers) {that enables the cooperation among multiple users watching multiple videos. 
			Using the demo, we further analyze the real-world performances of the CMS.}

		\item \emph{Experiments and Performances}: Based on the modified auction mechanism, 
		 we perform experiments in both simulative system and demo system. {Simulations with real traces show
			that crowdsourced mobile streaming outperforms noncooperative
			streaming by $48.6\%$ (on average) in social welfare. 
			{Experiments over the demo system further show that 
				those users who help others and those users who receive helps can increase their welfares by $15.5\%$ and $35.4\%$ (on average) via cooperation, respectively.} 
		}
	\end{itemize}

	
	The rest of this paper is as follows. \rev{We review related works in Section \ref{sec:literature}. We describe the system model in Section \ref{sec:model}, and propose incentive mechanisms in  Sections \ref{sec:single} and \ref{sec:multiple}. We further propose a modified mechanism in \ref{sec: modify}. Then, in Section  \ref{sec:demo}, we describe a demo system. In  Section \ref{sec:experiment}, we show experiment results. In Section \ref{sec:conclusion}, we conclude this work.}

	
	
	

\section{Related Work}\label{sec:literature}

\subsection{Adaptive Bitrate Streaming}\label{subsec:liter-abr}

Most of early studies on ABR focused on single-user  bitrate adaptation methods, such as buffer-based adaptation\cite{ABR-single-Spiteri16,ABR-single-Huang15}, bandwidth-based adaptation\cite{ABR-single-Li14}, and hybrid buffer-bandwidth adaptation\cite{ABR-single-Zhou16,ABR-single-Hao14,ABR-single-Yin15}. 

To better utilize the network resources, {some recent works studied multi-user streaming models, which can be divided into four types\cite{magazine}: \emph{D2D models}\cite{ABR-multi-Keller2012,ABR-multi-Li2015,ABR-multi-Cao14}, where users share their {downloaded} video segments with other users through D2D links; \emph{P2P models}\cite{ABR-P2P-Xu2013, ABR-P2P-Abdallah15,ABR-multi-Klusch14}, where users download video segments from other users who have already downloaded it \revh{through the Internet};  {\emph{BA models}\cite{ABR-multi-Zhong16,ABR-multi-Zhang2014,ABR-multi-Seenivasan2014}, where multiple users aggregate their bandwidth to serve one user's video streaming need; \emph{\rev{CMS system}}\cite{ABR-multi-Lin2016}, {where multiple video users (who may watch different videos) form groups to share their cellular resources to serve all users' video streaming  needs}. 
		
		We summarize the key features of these works in Table I. Specifically, from the model's perspective, we compare two features: multi-server, ``$\surd$" if  videos can be downloaded from multiple servers (users with downloaded videos can also be regarded as servers); multi-video, ``$\surd$" {if  users watch different videos}. From the method's perspective, we compare three features: multi-seg per allocation, ``$\surd$" if  multiple segments can be allocated in an allocation; bitrate adaptation, ``$\surd$" if  bitrate adaptation is considered; incentive, ``$\surd$" if incentive mechanism is considered. We also compare whether the studies involve real demonstration system or not. }
	In our earlier work \cite{ABR-multi-Lin2016}, we proposed a \rev{CMS system} and derived the corresponding offline optimization problem. 
	However, in practice, a properly designed incentive mechanism is always required to motivate user cooperation, as 
	cooperations might lead to additional costs. This motivates the study of incentive mechanism in this work.
	\subsection{Multi-Dimensional Auction}\label{subsec:liter-auction}
	A multi-dimensional auction  enables bidders to reveal multi-dimensional information regarding the auctioned goods, such as price and quality.
	Che proposed a multi-dimensional auction framework \cite{Yeon-KooChe1993}, based on which  Asker \emph{et al.} in  \cite{JohnAsker2008} and David \emph{et al.} in \cite{EstherDavid2005} studied auction properties under specific score functions. 
	As the multi-dimensional auction generalizes the single-dimensional auction, it has found wide applications in financial markets \cite{HGimpel2006} and power procurement \cite{HungPoChao2002}.~~~~~~~~~~~~~~~~~~~
	
	Most of the existing works on the multi-dimensional auction considered single-object allocation, where only one good is allocated in each auction. In \cite{MartinBichler2000Multi}, Bichler \emph{et al.} showed that the multi-object extension in multi-dimensional auction  is difficult because of the preferential dependence: bidders' preferences of the price depend on their preferences of the quantity. \revk{Specifically, with the preferential \revr{dependence}, the widely used score function in the \emph{additive} form (as in Definitions 1 and 3) \revr{fails to characterize the relationship} between the  price and the quality dimensions. If the score function is non-additive, the auction will be quite challenging to analyze.} In \cite{MartinBichler2000Multi}, the authors proposed a continuous auction mechanism in the multi-object case, without the guarantee of either truthful bidding or efficient resource allocation. In addition, \revk{the continuous auction is unsuitable for video streaming applications, because such an auction incurs a large signaling overhead in practice as bidders have to submit bids repeatedly to reach an agreement.} 
 	
 	\revk{In this work, instead of capturing  all three dimensions (i.e., price, bitrate, and quantity) in the score function, we only capture  price and bitrate dimensions using an additive form. We address the quantity dimension by enabling each bidder to submit a set of two-dimensional bids (bitrate and quality dimensions), each of which corresponds to the bid under a particular segment number (quantity dimension). This represents a new approach of  the multi-object multi-dimensional auction design. }
 	As far as we know, this is the first work that achieves truthful bidding and efficient resource allocation in a multi-object multi-dimensional auction.

\section{System Model} \label{sec:model}

In \rev{a CMS system}, we consider a set of mobile users  $\N \eq \{1,2,...,N\}$, who  download videos cooperatively. Each user watches \rev{a video that is encoded based on the  ABR technology} and is downloaded via cellular links to his  mobile device. 

\subsection{Adaptive Bitrate Streaming}\label{subsec:model-abr}
We consider a typical ABR streaming protocol\cite{abr} in the CMS system. 
Its key features are summarized as follows. 

\emph{\textbf{Video Segmentation:}} A source video is partitioned into a sequence of small segments, each of which contains a piece of the source video with a fixed  playback time  (e.g., 10 seconds). 

\emph{\textbf{Multi-Bitrate Encoding:}} A segment is encoded in multiple copies with different  bitrates, 
so that a user can select the most suitable bitrate for each segment. 
Such a bitrate selection can be based on many factors, such as real time network conditions and individual preferences. 

\emph{\textbf{Data Buffering:}} To smooth the playback, each downloaded segment is saved in a buffer at the user's device before playing.
The video player on the user's device fetches segments from the buffer sequentially for playback.
Due to the device's storage limit, the buffer has a limited maximum size. 

For a user $n$, let $\tseg_{n}>0$ denote his video's fixed segment length  (in terms of playback time), let $\R_n \eq \{R_n^1, R_n^2, ..., R_n^Z\}$ denote the corresponding finite bitrate set, 
and let $ \bufsize_n > 0$ denote his maximum buffer size (in terms of playback time).

%

\subsection{Crowdsourced Mobile Streaming}\label{subsec:model-cms} 
%
In \rev{the CMS system}, users who are close-by form a mesh network  and share their network resources.  Through a proper scheduling mechanism, the group of users cooperatively download the requested segments of the entire group through cellular links and then forward segments to the actual requesting users (receivers) through WiFi  Direct links. 
Different users can watch different videos in this framework. 

We consider a continuous time model over a period of time $\T \eq  [0,\ T]$, where $t = 0$ is the initial time and $t= T $ is the ending time. Let ${h}_{n}(t) > 0$ denote user $n$'s cellular link capacity at time $t\in\T$. 
Let $e_{n,m}(t) \in \{0,1\}$ denote the encounter between users $n$ and user $m$ at time $t$, i.e., $e_{n,m}(t)=1$ if user $n$ and user $m$ are encountered. Note that a user always encounters himself, i.e., $e_{n,n}=1$ for all $n$. 


\subsection{User Model} \label{subsec:model-use}

\revk{We first describe  the welfare generated through the downloading operation between two users. Then, we define the social welfare of the system, which is the sum of the welfare generated by all downloading operations.}
	
In the downloading operation between two users, a user $\downloader\in \N$ downloads a sequence of a total of $\kappa$  segments with bitrates $\brvector=\{\br_1,\br_2,...,\br_{\segamount}\}$ for a user $\receiver\in\N$, where  $\br_i >0$ for all $i$. User $\downloader$  and user $\receiver$ can be the same user. 
The downloading of segment $i$ starts at $\starttime_i$ and ends at $\endtime_i$. 
The downloading timings and the channel condition satisfy the following relationship:
\begin{equation}
	\int_{\starttime_{i}}^{\endtime_i} {h}_\downloader(t) \mathrm{d} t =  \br_i \cdot \tseg_\receiver , ~ i=1,2,...,\segamount,
\end{equation}
where the total downloaded volume within the downloading time is equal to the size of the downloaded segment. 

This downloading operation (by user $\downloader$ for user $\receiver$) induces a  cost for user $\downloader$ and a utility for user $\receiver$.

\subsubsection{\textbf{Cost of Downloader (User $n$)}}\label{subsec:mobel-user-cost}
Cost of the downloader is user $n$'s cost for downloading and transmitting video segments with bitrates $\brvector=\{\br_1,\br_2,...,\br_{\segamount}\}$. The cost function $\Cost_{n,t}(\brvector)$ 
consists of the cost on cellular links and the cost on WiFi Direct links:
the cellular cost is the cost for downloading the segments (e.g., energy consumption or cellular data payment), while  the WiFi Direct cost is the  cost that user $n$ transmits the segments to user $m$ if $n\neq m$ (e.g., energy consumption). 
Let $c_{n,t}(\br)$ be the {cellular and the WiFi Direct cost} for a single segment with bitrate $r$. We assume that  the  cost $c_{n,t}(\br)$ is a non-decreasing linear function, i.e., $c_{n,t}(0) = 0$, $[c_{n,t}(\br)]_r \geq 0$, and $[c_{n,t}(\br)]_{rr} = 0$.\footnote{{Let $[\cdot]_{x}$ and $[\cdot]_{xx}$ denote the first and the second order derivatives with respect to variable $x$, respectively. Let $[\cdot]_{xy}$ denote the second-order partial derivative with respect to variable $x$ and variable $y$.}} \revk{The linear model for downloading and transmission energies has been widely considered in the existing works on video services and user-provided networks \cite{energy2,energy3}, and the linear data payment is essentially a usage-based model commonly adopted by mobile network operators today. Relaxing the linearity  assumption will not affect the auction mechanism and its corresponding  properties, but will affect the particular form of the sufficient condition to satisify Assumption \ref{ass:multi-score} (Proposition \ref{cnd:non-negative}).} 
We assume that the cost of different segments are independent of each other, so the cost $C_{n,t}(\brvector)$ can be represented as follows:
\begin{equation}\label{eq:cost}
	\Cost_{n,t}(\brvector) = \sum_{i=1}^{\segamount}c_{n,t}(r_i).
\end{equation}  

\subsubsection{\textbf{Utility of Receiver (User $m$)}}\label{subsec:model-user-utility}
Utility of the receiver is user $m$'s utility for receiving $\segamount$ video segments with bitrates $\brvector=\{\br_1,\br_2,...,\br_{\segamount}\}$. A user often desires  to watch a high quality video without frequent video freezings (i.e., rebuffers) or quality degradations \cite{ABR-single-Hao14,ABR-single-Zhou16,ABR-single-Yin15}. Hence, the utility depends on three factors: the video quality gain, the buffer filling gain, and the quality degradation loss.  \revk{The buffer filling gain can be  used to predict  the rebuffering probability, since whether the  exact rebuffering will occur or not is unknown when users make downloading decisions.} \revk{Similar as in  \cite{degrade2}, we consider the quality degradation loss instead of the quality switching loss (that considers both the degradation and the upgrade losses),  as humans are more sensitive to the degradation\cite{degrade}.}

The utility function $\Utility_{m,t}(\brvector)$ 
is related to the receiver $m$'s desire  for a high quality video $\va_{m,t}$, his current buffer level $\curbuf_{m,t}$, and his previous segment bitrate $\prebr_{m,t}$. Formally,
\begin{equation}\label{eq:utility}
		\Utility_{m,t}(\brvector)\eq\Valueq(\brvector,\va_{m,t}) + \Valueb(\revk{\segamount},\curbuf_{m,t}) - \Lossdeg(\brvector,\prebr_{m,t}), 
	\end{equation}
\revr{
where the summation form of \eqref{eq:utility} is commonly used in existing ABR works, such as \cite{ABR-single-Yin15}.}

	\emph{a) Video Quality Gain $\Valueq(\brvector,\va_{m,t})$} is the user's gain in terms of the video segment quality. A user has a higher gain if he receives a segment with a higher bitrate. 
	The user-dependent factor $ \va_{m,t} $ reflects user $m$'s desire for a high quality video. Let $\valueq_m(r,\va_{m,t})$ be the video quality gain function for a single segment with bitrate $r$, and this gain function is non-decreasing and concave\rev{\cite{concaveutility}}, i.e., $[\valueq_m(r,\va_{m,t})]_r\geq 0$ and $[\valueq_m(r,\va_{m,t})]_{rr}\leq 0$. The quality gain is zero when the segment bitrate is zero (receives nothing), i.e., $\valueq_m(0,\va_{m,t}) = 0$, for any $\va_{m,t}$. Moreover, a user with a higher $\va_{m,t}$ has a higher desire to increase the bitrate, so $[\valueq_m(r,\va_{m,t})]_r$ is non-decreasing in $\va_{m,t}$, i.e., $[v_m^{\textsc{q}}(r,\va_{m,t})]_{r\theta}\geq 0$.
	Suppose that the quality gain of each segment is independent of that of the others, i.e., 
	\begin{equation}
		\Valueq_m(\boldsymbol{\br},\va_{m,t}) = \sum_{i=1}^{\segamount} \valueq_m({r}_{i},\va_{m,t}).
	\end{equation}
	
	\emph{b) Buffer filling gain $\Valueb(\revk{\segamount},\curbuf_{m,t})$} is the user's gain in terms of filling the playback buffer\cite{ABR-single-Spiteri16}, \revk{which is a gain related to the segment number $\segamount$ only}. A user will have a higher gain if he receives more segments in an allocation, as this leads to a  reduced chance of video freezing. \revk{The gain of each additional segment decreases in the number of segments. This is because if a user has already been allocated a larger number of segments, he is less willing to obtain an additional one \revr{due to the reduced probability of rebuffering.}}  
	For a user with a lower current buffer size, he is more willing to have new segments, so he will have a higher buffer filling gain for \revh{the allocated segments}.   Let $\curbuf_{m,t}$ denote user $m$'s real-time buffer level  at time $t$,
	which is measured in terms of the playback time. 
	For notation convenience, we define a buffer filling gain gap between total $\kappa+1$  segments and  total $\kappa$ segments under buffer level $\curbuf_{m,t}$, as \revk{$\Delta(\segamount,\curbuf_{m,t}) =\Valueb_{m}(\segamount+1,\curbuf_{m,t})-\Valueb_{m}(\segamount,\curbuf_{m,t})$.} Summarizing the above discussions, the function $\Valueb_{m}(\segamount,\curbuf_{m,t})$ satisfies the following inequalities: 
	\begin{equation}
		{[\Valueb_{m}(\segamount,\curbuf_{m,t})]_{\curbuf_{m,t}} \leq 0,}
	\end{equation}
		\begin{equation} \label{eq:buffer-Delta}
		\Delta(\segamount,\curbuf_{m,t})\geq 0, ~ \Delta(\segamount+1,\curbuf_{m,t})- \Delta(\segamount,\curbuf_{m,t}) < 0.
		\end{equation}

		
		\emph{c) Quality Degradation Loss $\Lossdeg(\brvector,R_{m,t}^{\textsc{pre}})$} is the user's loss when the video degrades from a higher bitrate to a lower bitrate. 
		The user  will have a higher degradation loss if the degradation gap is larger. Let $\lossdeg_m(\hat{r},{r})$ be the degradation loss due to the fact  that a newly downloaded segment degrades from the previous bitrate $\hat{r}$ (of the previous segment) to the current bitrate ${r}$ with the  gap $\Delta r = \hat{r}-{r}$:
		\begin{equation}
			\lossdeg_m(\hat{r},{r}) = \left\{\begin{array}{ll}
				0,& \hat{r}<{r},\\
				\llossdeg_m(\Delta r),& otherwise.
			\end{array}\right.
		\end{equation}
		\revk{The positive part $\llossdeg_m(\Delta r)$} linearly increases with $\Delta r$\cite{ABR-single-Yin15}, i.e.,
		\begin{equation}
			{[\llossdeg_m(\Delta r)]_{\Delta r}\geq 0,~~[\llossdeg_m(\Delta r)]_{\Delta r \Delta r}= 0.}
		\end{equation}
		Let $ r_0 = \prebr_{m,t}$ be the bitrate {of the segment that user $m$ receives immediately before the new downloading segments}. The loss $\Lossdeg_m(\brvector,\prebr_{m,t})$ of the segments with bitrates $\boldsymbol{r}$ is the sum of the degradation loss of all the segments. Formally,
		\begin{equation}
			\Lossdeg_m(\brvector,\prebr_{m,t}) = \sum_{i=1}^{\segamount}\lossdeg_m(r_{i-1},r_{i}).
		\end{equation}

		\subsubsection{\textbf{Social Welfare}}\label{subsec:model-user-social}
		In the downloading operation by user $\downloader$ for user $\receiver$, the generated welfare is defined as the difference between the user $m$'s utility and the user $n$'s cost:
		\begin{equation}
			\SW_{nm,t} (\boldsymbol{\br}) =   \Utility_{m,t}(\boldsymbol{\br})  - \Cost_{n,t}(\boldsymbol{\br}).
		\end{equation}
		
		The social welfare of the system is the sum of the welfares that are generated through all the downloading operations among all the users. 

		\subsection{Problem Formulation}\label{subsec:model-pro}
		We aim to design an incentive mechanism in \rev{the CMS system} that \revi{can reveal user's private information} and maximizes the social welfare. The mechanism should help each user to decide how to allocate the downloading opportunities of $K$ segments to near-by users: 
		\emph{(i) who is the receiver of each of the segments, (ii) what is the bitrate of each of the segments, and (iii) what is the payment of each of the segment receivers?} 
		
		We will design auction mechanisms to address the issue of user private information. More specifically, we  propose a single-object ($K=1$) and a multi-object ($K\geq 1$) auction-based mechanism in Section \ref{sec:single} and Section \ref{sec:multiple}, respectively.

\section{Single-Object Auction-Based Mechanism}
\label{sec:single}
We adopt an auction-based incentive mechanism, in which users allocate segment downloading opportunities using auctions. At each decision epoch of any user (who is ready  to  download segments), he acts as an auctioneer and initiates an auction for all nearby users for deciding the next $K$ segments ($K=1$ in this section) to be downloaded. This framework operates in an asynchronous and decentralized manner, in the sense that each user initiates an auction independently and asynchronously from other users. \revg{To clarify, a double auction (that involves multiple auctioneers in an auction) is not suitable for the CMS system. This is because if implementing a double auction, under the auctioneers' asynchronous downloading operations, the auctioneers who are ready for downloading earlier have to wait for those who are ready later, which can lead to a significant downloading resource waste.}

We propose a multi-dimensional auction based framework, in which the bidders reveal their \emph{intended bitrates} and \emph{intended prices} through submitting multi-dimensional bids on the $K$ segments to be downloaded. \revk{Without loss of generality, we consider an auction initiated by a downloader (auctioneer) $n$ to his encountered users. We assume that \revr{the auction period (including the auction operation period and the segment downloading period)} is short enough, so that user $n$'s encountered users do not change during the auction  period. \revr{Such an assumption is supported by the fact that a segment often has  a small size (e.g., 10  seconds), and the corresponding total downloading and device-to-device transmission time is relatively small (comparing with the user's mobility time scale).} 
Let $\N_n$ denote  the set of user $n$'s encountered users: 
	\begin{equation}
	\N_n \eq \{ m \in \N\ |\ e_{n,m}(t) = 1, t\in[t_0 , t_0 + \epsilon] \},
	\end{equation}
where $[t_0 , t_0 + \epsilon]$ is the auction period.} Let $|\N_n|$ denote the total number of users in set $\N_n$. Note that the downloader will also join the auction as a \emph{virtual} bidder (i.e., $e_{n,n}(t) = 1$) to fulfill his own service requirement. 
	The bidder $m$'s private information is his real-time utility function, i.e., $\Utility_{m,t}(\cdot)$, depending on his desire for high quality video $\va_{m,t}$, current buffer level $B_{m,t}^{\textsc{cur}}$, and previous segment bitrate $R_{m,t}^{\textsc{pre}}$ (Section \ref{subsec:model-user-utility}). 	{We also assume that these functions and parameters do not change during a single auction.} 

	In this section, we focus on single-object multi-dimensional auction mechanism design, where the auctioneer allocates one segment in each auction, i.e., $K=1$. 
	
	\subsection{Auction-Based Incentive Mechanism}\label{subsec:single-auction}
	\subsubsection{\textbf{SOMD Auction Framework}}
	In the SOMD auction framework, each bidder  submits a two-dimensional bid comprising \emph{bitrate} and \emph{price}. According to the bids, the auctioneer allocates the single segment to a single winner. Formally, 

		\begin{framework}
			[SOMD Auction Framework]\label{mech:auction-single}	~~
			\begin{enumerate}
				\item The auctioneer $n$ announces auction rules, including the  \emph{allocation rule} $\Win (\cdot) $ and  the \emph{payment rule} $\Pay  (\cdot)$;
				\item Each bidder $m \in {\N_n} $ submits a  bid $\bid^m=(\br^m,\pr^m)$  to maximize his own expected payoff. \revh{Let $\boldsymbol{b} = (\boldsymbol{b}^m, \forall m\in\mathcal{N}_n)$ denote the bids from all the bidders};
				\item The auctioneer $n$ determines the winner $\win$  and the winner's payment $\wpay$ according to the announced rules:
				\begin{equation}
					\win = \Win(\bid), \quad \wpay=  \Pay(\bid).
				\end{equation}
				The auctioneer will download a segment for winner $\win$ at the bitrate specified in the winner's bid, i.e, $\wbr = \br^{\win }$.
			\end{enumerate}
		\end{framework}
	\revk{Here,  $\wpay$ is the actual payment from the winner, which may not be equal to the price $\pr^{\win}$ submitted by the winner.} 
	Given an auction outcome $(\win, \wpay,\wbr)$, the auctioneer $n$'s payoff is 
	\begin{equation}
		P_n(\wpay, \wbr)=  \wpay- \Cost_{n,t}(\wbr),
	\end{equation}
	and the receiver (winner) $\win$'s payoff  is
	\begin{equation}
		P_{\win}(\wpay, \wbr)= \Utility_{\win,t}( \wbr) - \wpay  ,
	\end{equation}
	where $\Cost_{n,t}(\cdot) $ is the downloader's cost defined in \eqref{eq:cost},
	and $ \Utility_{\win,t}( \cdot)$ is the receiver's utility defined in \eqref{eq:utility}. 	
	\subsubsection{\textbf{Second-Score Auction}}\label{subsec:single-auction-2ndscore}
	
	The {winning rule} $\Win (\cdot) $ and the {payment rule} $\Pay  (\cdot)$ are two key elements in auction design.
	In a single-dimensional auction, the auctioneer can determine the winner by simply sorting all bidders' prices and choosing the bidder with the highest price.
	In a multi-dimensional auction here, however, the auctioneer cannot determine the winner by simply choosing the bidder with the highest price.
	This is because the bitrate of bidder will affect the auctioneer's downloading cost, and hence the auctioneer's payoff.
	
	To this end, we introduce a \emph{score function}  to determine the winner and the payment.
	The key idea is to transform a multi-dimensional bid $(\br,\pr)$ into a single-dimensional score $\phi(\br, \pr)$, so that the auctioneer can sort bidders according to  their scores and determine the winner by choosing the highest score bidder.
	{In this work, we adopt the class of score functions in \cite{EstherDavid2005}.}
	
	
	\begin{definition}[Single-Object Score Function]\label{def:score-single}
		{Given the bitrate $r^m$ and the price $p^m$ submitted by a bidder $m$}, the single-object score function is defined as 
		\begin{equation}\label{eq:score-single}
			\revk{\phi(\br^m, \pr^m) = \pr^m - s(\br^m),}
		\end{equation}
		where  $s(\cdot)$ is a  non-decreasing function with  $s(0)=0$.
	\end{definition}
	
	{Intuitively, such a score function increases with the bidder's price and decreases with the
		bidder's bitrate, capturing the fact that the auctioneer prefers
		a higher price and a lower bitrate.
		Note that \eqref{eq:score-single} corresponds to a class of score functions, as we have not specified the concrete form of function $s(\cdot)$. A key contribution of our work is to design the function $s(\cdot)$ properly  in order to  achieve desirable outcomes such as efficiency (social welfare maximization).}

	We implement a \emph{second-score} (multi-dimensional) auction \cite{Yeon-KooChe1993},
	where the winner is the bidder with the highest score, and the winner's payment is the price that ``derives"  the second highest score under the winner's bitrate. Formally,
	
	\begin{mechanism}[Second-Score Auction]\label{mech:2ndscore}
		The second-score auction is a special case of Framework \ref{mech:auction-single}, where the allocation and bidding rules are defined as follows:
		\begin{enumerate}
			\item \emph{Allocation Rule}: The winner $\win$ is the bidder with the highest score, i.e., 
			\begin{equation}
				\win= \arg \max_{m\in {\N}_n} \ \phi(\br^m, \pr^m);
			\end{equation}
			\item \emph{Payment Rule}: The winner's payment $\wpay$ is the price that derives the second highest score under his bitrate $\wbr = \br^{\win}$, i.e., 
			\begin{equation}
				\wpay  -  s(\wbr) = \max_{m\in {\N}_n/\win} \phi(\br^m, \pr^m). 
			\end{equation}
		\end{enumerate}
	\end{mechanism}
	
	
	{Next, we first analyze any bidder's \rev{optimal} price and bitrate strategies in Section  \ref{subsec:single-truth}. Based on the \rev{optimal} strategies, we propose an efficient mechanism  through a proper choice of score function in Section \ref{subsec:single-efficient}.}
	
	\subsection{{Truthfulness and Optimal Bitrate}}\label{subsec:single-truth}
	
	In the second-score auction, we will show that each bidder will submit his bid (i.e., price and bitrate) according to Proposition \ref{prop:single-truthful} and \ref{prop:single-bitrate} to maximize his expected payoff, with proofs in Appendix \ref{app:single-truthful} and \ref{app:single-bitrate}, respectively.
	
	
	\begin{proposition}[Truthfulness]\label{prop:single-truthful}
		Given any bitrate bidding strategy $\br^m$, the optimal price bidding strategy $\pr^m$ of a bidder $m$ is his true utility under the selected bitrate $\br^m$, i.e.,
		\begin{equation}\label{eq:single-truthful}
			\pr^m = \Utility_{m,t }( \br^m).
		\end{equation}
	\end{proposition}

	\begin{proposition}[Optimal Bitrate]\label{prop:single-bitrate}
		The optimal bitrate bidding strategy $\br^m$ of a bidder $m$ is given by
		\begin{equation}\label{eq:single-bitrate}
			\br^m = \arg \max_{\br\in \R_m }\left(\Utility_{m,t }( \br) - s(\br)\right) .
		\end{equation}
	\end{proposition}

	
	 Propositions \ref{prop:single-truthful} and \ref{prop:single-bitrate} propose the optimal strategy of each bidder $m$ in the second-score auction. \revk{To maximize his own profit, each bidder should select the bitrate $\br$ that maximizes the difference between his  utility $\Utility_{m,t}(\br)$ and the $s(r)$, and select the price $\pr$ that is equal to  his utility under the optimized  bitrate.} \revk{Due to the finite choices of bitrate, an optimal solution  always exists, and each bidder is able to calculate the optimal solution based on  his local information and the auctioneer's announced information with  a low computation complexity.}
	
	
	\subsection{Efficiency}\label{subsec:single-efficient}
	Notice that Propositions \ref{prop:single-truthful} and \ref{prop:single-bitrate} hold for any score functions  in the form of \eqref{eq:score-single}. On the other hand, the specific choice of $s({r})$ determines bidders' \rev{optimal} strategies and auction's allocation and payment, so auctioneers can choose the  score function to achieve desirable auction outcomes. 
	
	Here, we propose the efficient mechanism that maximizes the social welfare. 
	We first define an efficient score function:
	\begin{definition}[Single-Object Efficient Score Function]\label{df:single-efficient}
		An efficient score function is in the form of
		\begin{equation}\label{eq:single-efficient}
			\score(\br, \pr) = \pr - C_{n,t}(\br),
		\end{equation}
		where $C_{n,t}(\br)$ is the auctioneer's downloading cost.
	\end{definition}
	

	Under the score function of \eqref{eq:single-efficient}, we next show that the second-score auction implements the efficient mechanism.
	\begin{theorem}[Efficiency]\label{prop:single-efficient}
		Under the \rev{optimal} bidding behavior specified in Propositions \ref{prop:single-truthful} and \ref{prop:single-bitrate}, the second-score auction with the efficient score function in \eqref{eq:single-efficient} implements the efficient mechanism that maximizes the social welfare.
	\end{theorem}
	\begin{proof}
			Based on Proposition \ref{prop:single-truthful} and \ref{prop:single-bitrate}, each  bidder $m$ submits bid $(\br^m,\pr^m)$, where $\br^m = \arg \max_{\br\in \R_m }\left( \Utility_{m,t}( \br) - C_n(r)\right)$ and $\pr^m = \Utility_{m,t}(\br^m)$. \revh{In other words, each bidder submits the bitrate $r^m$ that maximizes his score. This leads to }
		\begin{equation}
		\phi(\br^m,\pr^m) = \max_{r\in \R_m} \left(\Utility_{m,t}(\br) - \Cost_n(\br)\right).
		\end{equation}
		In second-score auction, the winner $\sigma^\dag$ is the bidder with the highest score, i.e.,
		\begin{equation}
		\revh{  \sigma^{\dag}= \arg \max_{\substack{ m \in {\N}_n }} \phi(\br^m,\pr^m) .}
		\end{equation}
		\revh{The winning bitrate $r^{\dag}$ is the bitrate submitted by the winner $\sigma^{\dag}$, i.e., $r^{\dag}=r^{\sigma^{\dag}}= \arg \max_{\br\in \R_{\sigma^{\dag}} }\left( \Utility_{\sigma^{\dag},t}( \br) - C_n(r)\right)$. Hence, the social welfare under  $\sigma^{\dag}$ and $r^{\dag}$ is as follows:}
		\begin{equation}
		\revh{ \Utility_{\sigma^{\dag},t}(\br^{\dag}) - \Cost_n(\br^{\dag})=\max_{\substack{ m \in {\N}_n }}  \max_{r\in \R_m} \left(\Utility_{m,t}(\br) - \Cost_n(\br)\right),}
		\end{equation}which implies that the social welfare is maximized.
	\end{proof}
	\revk{Note that the exact downloading capacity is unknown beforehand, which leads to an unknown cost function $C_{n,t}(r)$ in \eqref{eq:single-efficient} when an auctioneer initiates an auction. Hence, in practice, an auctioneer needs  to estimate his downloading capacity based on his historical information using methods such as the one  in \cite{ABR-single-Li14}. The design and optimization of such an estimation is outside the scope of this paper. In later simulations, we assume that an auctioneer calculates his cost function based on the average capacities of his previous several downloading operations. Although the estimation accuracy will affect the mechanism performance, bidders and auctioneers make decisions based on not only the cost $C_{n,t}(r)$ but also bidders' utilities $U_{m,t}(r)$ (which involves bidders' buffer level information). Hence, under the extreme case where  capacities vary dramatically, the consideration of the buffer levels can alleviate the performance degradation caused by inaccurate estimation.}

\section{Multi-Object Auction-Based Mechanism}\label{sec:multiple}

\rev{To reduce the possibly excessive signaling  overhead caused by the frequently auctions}, 
in this section, we consider the more general case of multi-object multi-dimensional auction mechanism design, where the auctioneer allocates multiple segments in each auction, i.e., $K\geq1$.  For nontation simplicity, we will write the bidder set $\N_n$ as $\M=\{1,2,...,M\}$, where $M$ is the total number of bidders in the set $\N_n$.}

\subsection{Auction-Based Incentive Mechanism}\label{subsec:multi-auction}

\subsubsection{\textbf{MOMD Auction Framework}}

In the MOMD auction framework, bidders submit multi-dimensional bids, revealing their intended \emph{bitrate} and \emph{price} under each segment that might be allocated. Based on the bids, the auctioneer allocates the (downloading opportunities of) $K$ segments to multiple bidders. An MOMD auction operates as follows:
\begin{framework}\label{mech:multi-auction}	
	[MOMD Auction Framework]~~
	\begin{enumerate}
		\item The auctioneer $\auctioneer$ announces auction rules, including the  \emph{segment number} $\objamount$,  the \emph{allocation rule} $\Win (\cdot) $, and the \emph{payment rule} $\Pay  (\cdot)$; 
		\item Each bidder $\bidder \in \M $ submits a bid $\bid^\bidder=(\brmatrix^\bidder,\prvector^\bidder)$  to maximize his own expected payoff. \revh{Let $\boldsymbol{b} = (\boldsymbol{b}^m, \forall m\in\M)$ denote the bids from all the bidders.} Here,
		\begin{itemize}
			\item Bitrate matrix
			\begin{equation}\label{eq:muli-bitratematrix}
				\brmatrix^\bidder  =
				\left[ \begin{array}{c}
					\brvector_{1}^\bidder \\
					\brvector_{2}^\bidder   \\
					\vdots  \\
					\brvector_{\objamount}^\bidder \\
				\end{array} \right] =
				\left[ \begin{array}{cccc}
					\br_{11}^\bidder & 0      & ...    & 0      \\
					\br_{21}^\bidder & \br_{22}^\bidder & ...    & 0      \\
					\vdots & \vdots & \ddots & \vdots    \\
					\br_{\objamount1}^\bidder & \br_{\objamount2}^\bidder & ...    &  \br_{\objamount\objamount}^\bidder\\
				\end{array} \right],
			\end{equation}
			
			where $r_{\kappa i}^\bidder \in \boldsymbol{R}^\bidder$ is the  bitrate  of the $i^{th}$ segment when bidder $\bidder$ is allocated a total of $\kappa$ segments.
			\item Price Vector
			\begin{equation}\label{eq:multi-pricevectore}
				\prvector^\bidder =
				\left(
				\pr_{1}^\bidder, \pr_{2}^\bidder, ..., \pr_{\objamount}^\bidder\right),
			\end{equation}
			where $p_{\kappa}^\bidder$ is the total price (willingness-to-pay)  when bidder $\bidder$ is allocated a total of $\kappa$ segments.
		\end{itemize}
		
		\item The auctioneer $\auctioneer$ determines the \emph{allocation set}, i.e., the winner of \emph{each segment}, $\winset\eq \{\win_1,\win_2,...,\win_\objamount\},$ and the \emph{payment set}, i.e., the price that \emph{each bidder} needs to pay, $\wpayset \triangleq \{\wpay_1,\wpay_2,...,\wpay_M\},$ according to the rules:
		\begin{equation}
			\winset= \Win(\bid), \quad \wpayset =  \Pay(\bid).
		\end{equation}
		Accordingly, the downloading  bitrate of each segment is equal to the  submitted bitrate of the  corresponding winner, denoted by $\wbrset \triangleq \{\wbr_1,\wbr_2,...,\wbr_\objamount\}$.
	\end{enumerate}
\end{framework}

Notice that both the allocation set and the bitrate set have the size of $K$, as these two sets enumerate the receiver and the bitrate for \emph{each segment}, respectively; however, the size of  the payment set is $M$, and each element corresponds to  the payment from \emph{a bidder}. To facilitate the later discussions, we define a revised allocation set $\winsetupdate$ and a revised bitrate set $\wbrsetupdate$, both of which have the size of $M$. More specifically, starting from allocation set $\winset$,  we can compute the number of segments allocated to bidder $m$, denoted as $\wamount_\bidder$. 
With this we can define the revised allocation set as $\winsetupdate=\{\wamount_{1},\wamount_2,...,\wamount_M\}$, where $\sum_{m=1}^M\wamount_{\bidder}=\objamount$. Similarly, we  define the revised bitrate set as
$\wbrsetupdate = \{\wbrupdate_1, \wbrupdate_2,...,\wbrupdate_M\}$,
where vector $\wbrupdate_\bidder$ is the bitrate set for the  $\wamount_m$ segments  allocated to bidder $m$, 
i.e., $\wbrupdate_\bidder = \brvector^\bidder_{\wamount_\bidder}$ (i.e., the $\wamount_\bidder$th row of bitrate bid matrix $\brmatrix^\bidder$). 

Based on the auction results, the auctioneer $\auctioneer$'s payoff is the sum of the difference between each bidder's payment and $n$'s downloading cost for this bidder's segments, i.e.,
\begin{equation}
	\payoff_{\auctioneer}(\wpayset,\wbrsetupdate) = \sum_{\bidder=1}^M [\wpay_\bidder - \Cost_{n,t}( \wbrupdate_\bidder )].
\end{equation}
Bidder $\bidder$'s payoff is the difference between his utility and his payment, i.e.,
\begin{equation}
	\payoff_{\bidder}(\wpay_\bidder,\wbrupdate_\bidder) = \Utility_{m,t}(\wbrupdate_\bidder) - \wpay_\bidder. 
\end{equation}

\subsubsection{\textbf{Vickrey-Score Auction}}
In a multi-dimensional auction, the vector bids may not be sorted easily, \rev{and this introduces difficulties for determining the allocation set and the payment set.} We again introduce a \emph{score function} to address this problem. Different from single-object case in Section \ref{sec:single}, here we will  transform the bids into sequences of \emph{marginal scores}, of which the auctioneer can sort  and make decisions. 

We first define the score function as follows. 
\begin{definition}[Multi-Object Score Function]\label{df:multi-score}	
	\revk{Given the bitrate $\boldsymbol{R}^m$ and the price $\boldsymbol{p}^m$ submitted by a bidder $m$}, for any number of allocated segments $\kappa$, the multi-object score function $ \score(\brvector_\kappa^m, \pr_\kappa^m) $ is given by
	\begin{equation}\label{eq:score2}
		\score(\brvector_\kappa^m, \pr_\kappa^m) = \pr_\kappa^m - s(\brvector_\kappa^m),
	\end{equation}
	where $s(\cdot)$ is a component-wise  non-decreasing function and {$s(\boldsymbol{0})=0$}.
\end{definition}


The score function in \eqref{eq:score2} involves one row in the bitrate matrix in \eqref{eq:muli-bitratematrix} and one component in the price vector in \eqref{eq:multi-pricevectore}. Hence, for each bidder $m$, we will compute $K$ scores, i.e., $\score(\brvector_\kappa, \pr_\kappa), \forall \kappa = 1, ..., K$. Based on this, we can further compute  the \emph{marginal score sequence} for each bidder $m$: $\boldsymbol{S}^m =\{S_1^m,S_2^m,...S_K^m\},$
{where the $\kappa^{th}$ marginal score reflects  bidder $m$'s score increase when the total allocated segment number to bidder $m$ increases from $\kappa-1$ to $\kappa$.} Formally,
\begin{equation}\label{eq:marginalscore}
	\mscore_{\kappa}^m = 
	\left\{\begin{array}{ll}
		\score(\boldsymbol{\br}^m_1, \pr^m_1),&\kappa=1,\\
		\score(\boldsymbol{\br}^m_\kappa, \pr^m_\kappa) - \score(\boldsymbol{\br}^m_{\kappa-1}, \pr^m_{\kappa-1}),&2\leq \kappa \leq K. 
	\end{array}\right.
\end{equation}

\revr{We impose the following assumption on marginal scores:} 
\begin{assumption}[Marginal Score]\label{ass:multi-score}
	For any bidder $\bidder \in \M$, the marginal score sequence $\boldsymbol{\mscore}^\bidder$ is non-negative and non-increasing in $\kappa$, where:
	\begin{equation}
		\mscore^{\bidder}_{\kappa} \geq \mscore^{\bidder}_{\kappa+1} \geq 0, \quad \kappa = 1,2,...,\objamount-1.
	\end{equation}
\end{assumption}
\revr{Assumption \ref{ass:multi-score} implies that  an additional segment induces a larger score (i.e., a positive marginal score $\mscore^{\bidder}_{\kappa+1} \geq 0$), and  the score increase  (i.e., the marginal score) is non-increasing with the allocated segment number $\kappa$ (i.e., $\mscore^{\bidder}_{\kappa} \geq \mscore^{\bidder}_{\kappa+1}$). In Section \ref{subsec:multi-condition}, we provide a sufficient condition under which Assumption \ref{ass:multi-score} is always satisfied.}

Inspired by the VCG mechanism\cite{VCG1}, we propose a Vickrey-score auction, where we allocate  the $\objamount$ segments to the $\objamount$ highest marginal scores, and choose the payments reflecting the score damages of the winners to the system. {Next we will first define the proposed mechanism,  and then provide a numerical illustrating example. }

For a  bidder $\bidder$, let sequence $\mscoresetex^{-\bidder}$ denote the $\objamount$ highest marginal scores \emph{except} bidder $\bidder$'s:
\begin{equation}\mscoresetex^{-\bidder} \triangleq \{\mscoreex_1^{-\bidder},\mscoreex_2^{-\bidder},...,\mscoreex_\objamount^{-\bidder} \},\end{equation}
where $\mscoreex^{-\bidder}_k$ is the {$k^{th}$ highest value among all the bidders' marginal scores  \emph{except} bidder $\bidder$'s}. We further let $\mscoreseti$ denote the $K$ highest marginal scores among all bidders:
\begin{equation}\mscoreseti \triangleq \{\mscorei_1,\mscorei_2,...,\mscorei_\objamount\},\end{equation} 
where $\mscorei_k$ is the $k^{th}$ highest value among all the bidders' marginal scores. The Vickrey-score auction is as follows:
\begin{mechanism}[Vickrey-Score Auction]\label{mech:Vickrey-score}~
	The Vickrey-score auction is a special case of Framework \ref{mech:multi-auction}, where the allocation and payment rules are defined as follows:
	\begin{itemize}
		\item Allocation Rule: The segment $\obj$'s receiver $\win_\obj$ is the bidder corresponding to   the  $\obj^{th}$ highest marginal score, i.e.,
		\begin{equation}
			{\mscore_{i}^{\win_\obj} = \mscorei_\obj,	}
		\end{equation}
		where $\mscore_{i}^{\win_\obj}$ refers to the $i^{th}$ marginal score of bidder $\win_\obj$.
		\item Payment Rule: If bidder $m$ wins $\wamount_{\bidder}$ segments, then his payment $\wpay_m$  corresponds to the score damage caused by this bidder under his submitted bitrate, i.e.,
		\begin{equation}\label{eq:payment}
			\wpay_\bidder -  s(\brvector^\bidder_{\wamount_{\bidder}}) = \sum_{i=1}^{\wamount_{\bidder}}  \mscoreex_{K-\wamount_{\bidder}+i}^{-\bidder} .
		\end{equation}
	\end{itemize}
\end{mechanism}
{
	\begin{example}
		Consider an auction  with $M=3$ users and $K=4$  segments, where we have the following  marginal score sequences: $\mscoreset^1=\{ 8,~7,~5,~2\}$,	$\mscoreset^2=\{9,~6,~3,~2\}$, and $\mscoreset^3=\{4,~4,~3,~1\}$.
		Hence, we have the sorted sequences:
		\begin{equation*}
			\begin{aligned}
				\mscoreseti = \{9,~8,~7,~6\}; & 
				~~\mscoresetex^{-1} =  \{9,~6,~4,~4\};\\
				\mscoresetex^{-2} =  \{8,~7,~5,~4\}; & ~~\mscoresetex^{-3} =  \{9,~8,~7,~6\}.
			\end{aligned}
		\end{equation*}
		The four numbers in vector $\mscoreseti$ corresponds to the marginal scores of user 1 ($8$ and $7$) and user 2 ($9$ and $6$). Hence, according to the proposed  Vickrey-score auction: user $1$ wins two segments, and user $2$ wins two segments. The payments of user 1 and user 2 are:
		\begin{equation*}
			\begin{aligned}\wpay_1 = \sum_{i=1}^{2}  \mscoreex_{4-2+i}^{-1}+ s(\brvector^1_2) = \underbrace{4+4}_{score~damage}+ s(\brvector^1_2);\\
				\wpay_2 = \sum_{i=1}^{2}  \mscoreex_{4-2+i}^{-2}+ s(\brvector^2_2) = \underbrace{5+4}_{score~damage}+ s(\brvector^2_2).
			\end{aligned}
		\end{equation*}
		{Take user $1$ as an example: without user $1$, user $3$ will win $2$ segments with scores $4$ and $4$, so these scores are the score damage caused by user $1$. Hence,  user $1$ has to pay the price that compensates this damage as shown above. }
		
	\end{example}}  
	
	


	\subsection{Truthfulness and Optimal Bitrate}\label{subsec:multi-truth}
	
	
	In the Vickrey-score auction, {we prove that each bidder will submit his bid (i.e., price and bitrate) according to Proposition \ref{prop:multi-truthful} and Proposition \ref{prop:multi-bitrate} to maximize his expected payoff, with proofs in Appendix \ref{app:multi-truthful} and \ref{app:multi-bitrate}, respectively. 
		\begin{proposition}[Truthfulness]\label{prop:multi-truthful}
			Given any bitrate matrix $\brmatrix^\bidder$, the \rev{optimal}  price vector $\prvector^\bidder$ of a bidder $m$ is his true utility under the selected bitrate matrix $\brmatrix^\bidder$, i.e.,
			\begin{equation}\label{eq:multi-truthful}
				\pr^\bidder_{\kappa} = \Utility_{\bidder,t }( \brvector^\bidder_{\kappa}) , \quad {\kappa} = 1,2,...,K.
			\end{equation}
		\end{proposition}
		\begin{proposition}[Optimal Bitrate]\label{prop:multi-bitrate}
			For any number  of allocated segments $\kappa$ to bidder m, the \rev{optimal} bitrate vector $\brvector^\bidder_{\kappa}$ is the optimal solution $\brvector^{\star}$ of the following optimization problem:
			\begin{equation} \label{eq:optimal}
				\begin{aligned}
					& \underset{\brvector}{\text{maximize}}
					& & \Utility_{m,t}(\brvector) - s(\brvector) \\
					& \text{subject to}
					& & \br_{i}>0,  \quad i=1,...,\kappa,\\
					&&&\br_{i}=0, \quad i=\kappa+1,...,K,\\
					& \text{variables}
					& & {\br_i \in \R_m, \quad i=1,...,\kappa.}
				\end{aligned}
			\end{equation}
			The constraints restrict the allocated segment number to be $\kappa$.
		\end{proposition}
		\subsection{{Efficiency}}\label{subsec:multi-efficient}
		
		In this section, we propose the efficient score function that maximizes the social welfare. 
		\begin{definition}[Multi-Object Efficient Score Function]\label{df:multi-efficient}
			An efficient score function is in the form of
			\begin{equation}\label{eq:efficient}
				\score(\boldsymbol{\br}, \pr) = \pr - C_{n,t}(\boldsymbol{\br}),
			\end{equation}
			where $C_{n,t}(\brvector)$ is the downloading cost of the auctioneer.
		\end{definition}
		
		
		
		If each bidder submits the bid based on the \rev{optimal} price in Proposition \ref{prop:multi-truthful} and the \rev{optimal} bitrate in Proposition \ref{prop:multi-bitrate},  we prove that the Vickrey-score auction with the efficient score function maximizes the social welfare.
		\begin{theorem}[Efficiency]\label{prop:multi-efficient}
			Under the \rev{optimal} bidding behavior specified in Propositions \ref{prop:multi-truthful} and \ref{prop:multi-bitrate}, the Vickrey-score auction with the efficient score function in \eqref{eq:efficient} implements the efficient mechanism that maximizes the social welfare.
		\end{theorem}
		\begin{proof}
				\revk{In the Vickrey-score auction with an efficient score function, when bidding according to Proposition \ref{prop:multi-truthful} and \ref{prop:multi-bitrate}, any bidder $m$'s bid will induce a score $\phi^{m,n}_{\kappa}$ for  being allocated $\kappa$ segments, i.e., 
				$\phi^{m,n}_{\kappa}=\max_{\boldsymbol{r}_{\kappa}}\left( U_{m,t}(\boldsymbol{r}_{\kappa})-C_{n,t}(\boldsymbol{r}_{\kappa})\right) 
				$, where $\boldsymbol{r}_{\kappa}$ denotes the bitrate vector that satisfies the constraint of  $\kappa$ segments. Here,  $\phi^{m,n}_{\kappa}$ is essentially the maximum welfare that can be generated through the downloading by auctioneer $n$ for bidder $m$ under a particular segment number  $\kappa$.
				Let $\boldsymbol{\sigma}=\{{\kappa}_1,{\kappa}_2,...,{\kappa}_M\}$ denote an allocation set, where ${\kappa}_m$ is the number of segments  allocated to bidder $m$. In the Vickrey-score auction,  the auctioneer chooses the allocation set 
				$\boldsymbol{\sigma}^* = \arg \max_{\boldsymbol{\sigma}}\sum_{m=1}^{M} \phi^{m}_{{\kappa}_m} $, i.e., picking the set of allocation that maximizes the welfare generated between aunctioneer $n$ and bidders, that is, the social welfare.}
		\end{proof}
		{Finally, we comment on the applicability of the proposed Vickrey-score auction in existing video streaming systems, where the bitrate adaptation method has been specified. In this case, if each bidder chooses the bidding price according to  Proposition \ref{prop:multi-truthful} and use an existing bitrate adaptation method (e.g.\cite{ABR-single-Spiteri16,ABR-single-Huang15,ABR-single-Li14,ABR-single-Zhou16,ABR-single-Hao14,ABR-single-Yin15}),  the Vickrey-score auction with the  efficient score function is conditionally  efficient.
			\begin{corollary}[Conditional efficiency]\label{prp:subefficient}
				Given any fixed  bitrate $\brmatrix^\bidder$  for bidder $m$, Vickrey-score auction with the efficient score function maximizes the social welfare under the fixed  bitrates.
			\end{corollary}
		
			The proof of Corollary \ref{prp:subefficient} is similar as the proof of Theorem \ref{prop:multi-efficient} and hence is omitted.

\subsection{Conditions \rev{for Satisfying Assumption \ref{ass:multi-score}}}\label{subsec:multi-condition}
By now we have proved several desirable properties of the  Vickrey-score auction under Assumption \ref{ass:multi-score}. In this section, we will specify sufficient conditions, under which Assumption \ref{ass:multi-score} is satisfied. As an example, we will focus on  the efficient score function in \eqref{eq:efficient} in the rest of the discussions. Our discussions can also be generalized to other choices of score functions.
						

					
The rest of this subsection is divided into two parts. First, we prove some additional properties of a bidder's optimal bitrate matrix. Next, we  characterize sufficient conditions of the cost function $\Cost_{n,t}(\cdot)$ and the utility function $\Utility_{m,t}(\cdot)$ (defined in \ref{subsec:model-use}) in Proposition \ref{cnd:non-negative},  under  which  Assumption \ref{ass:multi-score} is satisfied.

Starting from  Proposition \ref{prop:multi-bitrate}, we prove that a bidder's optimal bitrate matrix has two features, as shown in Lemma \ref{lem:bitrate-p1} and \ref{lem:bitrate-p2}. Note that both lemmas are based on the efficient score function in \eqref{eq:efficient}, where the optimal bitrate vector in each row $\kappa$ is given:
\begin{equation}\boldsymbol{r}_{\kappa}^m = \arg\max_{\boldsymbol{r}_{\kappa}} \left(\Utility_{m,t}(\brvector_{\kappa}) - \Cost_{n,t}(\boldsymbol{\br}_{\kappa})\right).
\end{equation}Here, $\boldsymbol{r}_{\kappa}=\{r_{\kappa 1}, r_{\kappa 2},...,r_{\kappa\kappa}\}$ denotes the vector with $\kappa$ non-zero elements, and  $U_{m,t}(\boldsymbol{r}_{\kappa})= \Valueq_{m,t}(\boldsymbol{r}_{\kappa}) + \Valueb_{m,t}({\kappa})- \Lossdeg_{m,t}(\boldsymbol{r}_{\kappa},\prebr_{m,t})$. For presentation convenience, we define a function $g_{mn,t}(r) = \valueq_{m,t}(r) - c_{n,t}(r)$. Since the value of $\Valueb_{m,t}(\kappa,\curbuf_{m,t})$  depends on segment number $\kappa$ but not the value of $\boldsymbol{r}_{\kappa}$, the optimal vector $\boldsymbol{r}_{\kappa}$ can also be represented as:
\begin{equation}\label{eq:cnd-bitrate-g}	\boldsymbol{r}_{\kappa}^m=\arg\max_{\boldsymbol{r}_{\kappa}} \left(\sum_{i=1}^{\kappa } g_{mn,t}(r_{\kappa i})  -\Lossdeg_{m,t}(\boldsymbol{r}_{\kappa},\prebr_{m,t})\right).
\end{equation}
Any bidder's optimal bitrate matrix has the following features:
\begin{lemma}[\revh{Identical Bitrate}] \label{lem:bitrate-p1}
Under the  efficient score function in \eqref{eq:efficient}, any bidder $\bidder$'s optimal bitrate matrix $\brmatrix^m$ satisfies that in any row $\kappa$, the non-zero bitrate elements $r_{\kappa i}^{m}$ ($i\leq \kappa$)  are identical, hence can be written as $r_{\kappa i}^\bidder = r_{\kappa}^\bidder, ~ \forall  i\leq \kappa.$
\end{lemma}
\revh{The detailed proof is given in Appendix \ref{app:identical}, and here are the intuitions.  First, in \eqref{eq:cnd-bitrate-g}, the order of the non-zero elements in vector $\boldsymbol{r}_{\kappa}^m$ only affects function $ \Lossdeg_{m,t}(\boldsymbol{r},\prebr_{m,t})$, which is minimized when the elements are in the ascending order. Hence, the non-zero elements in the optimal vector $\brvector^m_{\kappa}$ has be in the ascending order. This means that the  bitrate degradation may only happen at the first segment, i.e.,
\begin{equation}\label{eq:cnd-bitrate-g3}\boldsymbol{r}_{\kappa}^m =\arg \max_{\boldsymbol{r}_{\kappa}}\left(\sum_{i=1}^{\kappa} g_{mn,t}(r_{\kappa i}) -\lossdeg(\prebr_{m,t},r_{{\kappa}1})\right).\end{equation} Second, there always exists an optimal bitrate (denoted by $r^*$) that maximizes $g_{mn,t}(r)$. If $\prebr_{m,t}\leq r^*$, then $r_{\kappa i}^m=r^*$ for all $i=1,2,...,\kappa$. If $\prebr_{m,t}\geq r^*$, \revi{we can obtain $r_{\kappa 1}^m\geq r^*$ by checking the partial derivate of the objective function in \eqref{eq:cnd-bitrate-g3} with respect to $r_{\kappa 1}$}. Moreover, \revi{the concave function $g_{mn,t}(r)$ is non-increasing with $r$ for $r\geq r_{{\kappa}1}^m \geq r^*$}, so \revi{$g_{mn,t}(r_{{\kappa}2}^m),...,g_{mn,t}(r_{{\kappa}\kappa}^m)$ are maximized when bitrates $r_{\kappa2},...,r_{\kappa\kappa}$ are minimized under the constraint that $r_{{\kappa}1}^m\leq r_{{\kappa}2}^m\leq... \leq r_{{\kappa}{\kappa}}^m$}, which implies $r_{{\kappa}1}^m =r_{\kappa 2}^m=... = r_{{\kappa}{\kappa}}^m $.}
\begin{lemma}[\revh{Non-Increasing Bitrate}] \label{lem:bitrate-p2}
Under the efficient score function in \eqref{eq:efficient}, any bidder $\bidder$'s optimal bitrate matrix $\brmatrix^m$ satisfies that the bitrate $r_{\kappa}^m$ for row $\kappa$ defined in Lemma \ref{lem:bitrate-p1} is non-increasing in the row index  $\kappa$: $r_{\kappa}^\bidder \geq r_{\kappa+1}^\bidder$ for all $\kappa = 1,2,...,K-1.$
\end{lemma}
\revh{According to Lemma \ref{lem:bitrate-p1}, the optimal common non-zero bitrate $r_{{\kappa}}^m$ for each row $\kappa$ is derived as follows: \begin{equation}\label{eq:lemma-nonincreasing} r_{{\kappa}}^m=\arg\max_r \left({\kappa}\cdot g_{mn,t}(r) -\lossdeg(\prebr_{m,t},r)\right).\end{equation} Intuitively, if $\prebr_{m,t}\leq r^*$, $r_{i}^m=r^*$ for all $i=1,2,...,\kappa$. If $\prebr_{m,t}\geq r^*$, as $\kappa$ increases, the impact of $\kappa\cdot g_{mn,t}(r)$ on the optimization problem increases, so $r_{\kappa}^m$ gradually decreases in $\kappa$ to approach $r^*$. The detailed proof is in Appendix \ref{app:non-increasing}. }
						
						
Based on Lemma  \ref{lem:bitrate-p1} and \ref{lem:bitrate-p2}, we show the sufficient conditions of $\Cost_{n,t}(\cdot)$ and  $\Utility_{m,t}(\cdot)$ for satisfying Assumption \ref{ass:multi-score}.
						
\begin{proposition}[Sufficient Conditions for Assumption 1]\label{cnd:non-negative}
		{The marginal scores are  non-negative for all $ m, n,t$, if}
\begin{equation}\label{eq:non-negative}
		v^{\textsc{q}}_m(r,\theta)  \geq c_{n,t}(r), ~\forall r, \theta.
\end{equation}
{The marginal scores are  non-increasing \revh{in $\kappa$ (i.e., the number of allocated segments)} for all  $ m,n,t$, if}
\begin{equation}\label{eq:non-increasing}
\begin{array}{r}
2K\cdot c_{n,t}( R_m^Z) + \lossdeg_m( R_m^Z,0)\leq |\tilde{\Delta}|,
\end{array}
\end{equation}
\revk{where $\tilde{\Delta}$ is the minimum value that satisifies} 
\begin{equation} 
\revk{0>\tilde{\Delta} \geq \Delta(\segamount+1,\curbuf_{m,t})- \Delta(\segamount,\curbuf_{m,t}), ~ \forall \segamount, \curbuf_{m,t}.}
\end{equation}
\end{proposition}
Intuitively, to satisfy  Assumption \ref{ass:multi-score}, the video quality gain of each allocated segment should be no less than the downloading cost of that segment to ensure the non-negative marginal scores, and the buffer filling gain should be concave enough (i.e., $|\tilde{\Delta}|$ should be large enough) to ensure the non-increasing marginal scores. The detailed proof is shown in Appendix \ref{app:condition}. 

\section{Mechanism Modification}\label{sec: modify}
	
				\begin{figure}[t]
		\center
		\includegraphics[height=3.3cm]{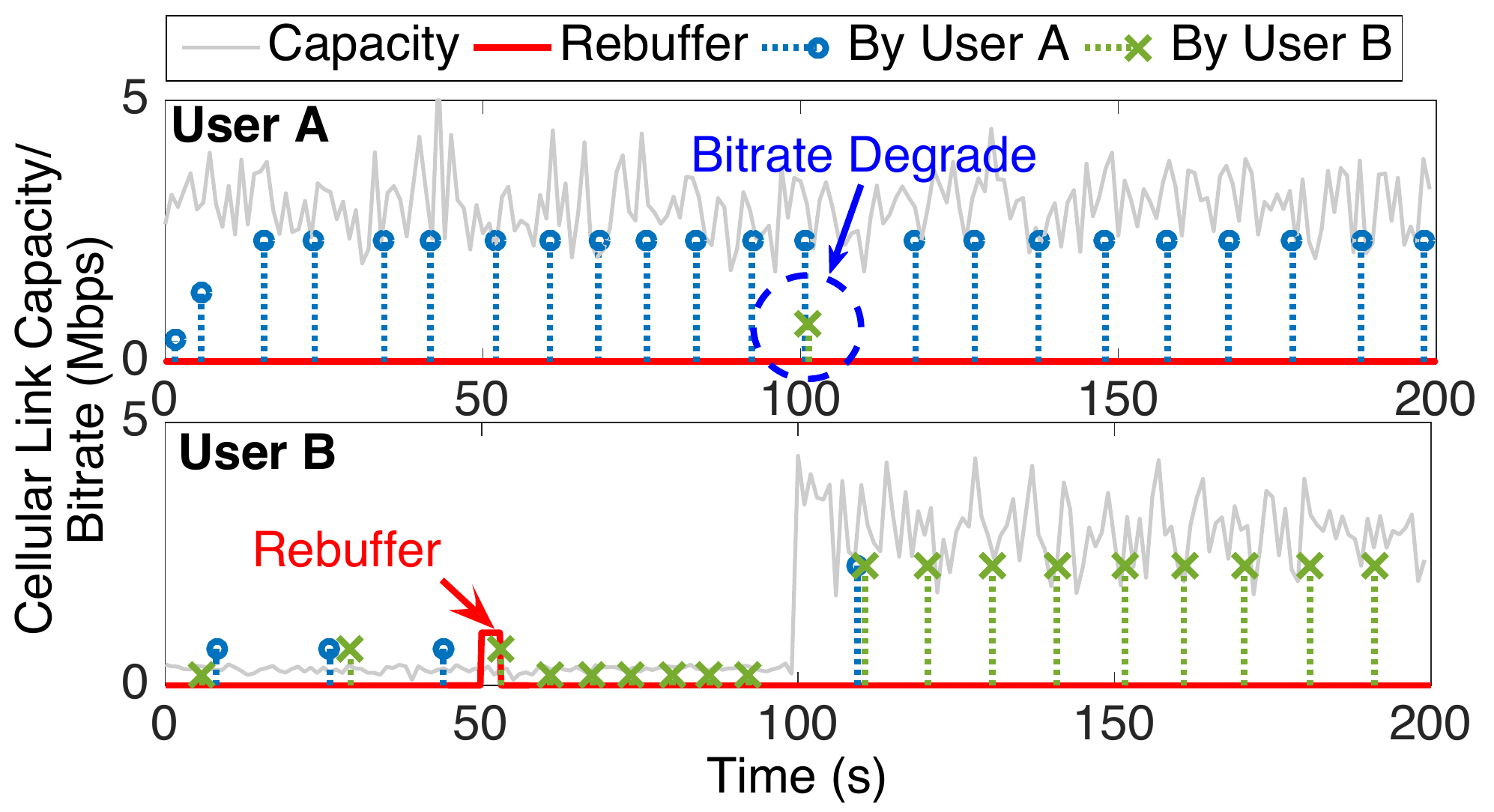}\\
		\caption{\revk{The Video Streaming of Users A and B under Mechanism \ref{mech:Vickrey-score}.}}\label{fig-modify:before}
	\end{figure}	
	\begin{figure}[t]
		\center
		\includegraphics[height=3.3cm]{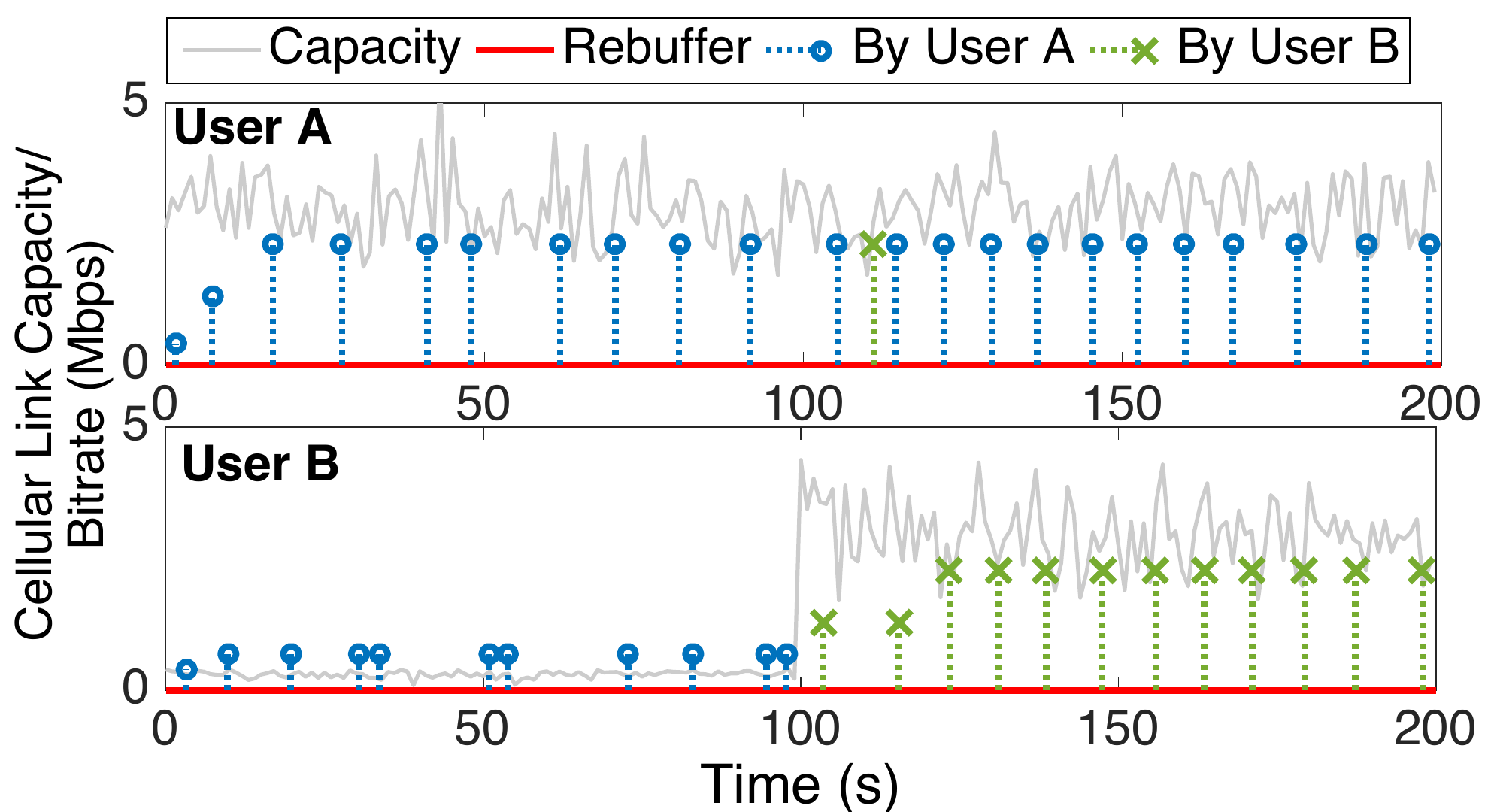}\\
		\caption{\revk{The Video Streaming of Users A and B under Mechanism \ref{mech:modification}.}}\label{fig-modify:after}
	\end{figure}
		
		%
		In Sections \ref{sec:single} and \ref{sec:multiple}, we proposed two auction-based incentive mechanisms for single segment and multiple segments downloading, respectively. By implementing the  efficient score function, the  mechanisms can maximize the social welfare in each auction. 
		However, since the social welfare maximization is performed in each auction independently, the long-term social welfare across multiple rounds of auctions may not necessarily achieve the maximum in some cases.  
		
		One  scenario worth considering is where the link capacities of some users are substantially poorer than others. Hence utilizing the downloading opportunities of these users might actually hurt the overall performance.  Figure \ref{fig-modify:before} shows the video scheduling processes of such a scenario with two users: user A and user B. Here $x$-axis corresponds to the video streaming time horizon (of 200 seconds), and $y$-axis corresponds to cellular network capacity (for the gray continuous curves). User  A has an average capacity of 3Mpbs along the whole streaming interval (200 seconds), while user B has an average capacity of 0.3Mbps during the first 100 seconds and an average capacity of 3Mbps during the latter 100 seconds. \rev{The stems with circles and crosses  are the segments that are downloaded by user A and user B, respectively, and the heights of these stems represent the corresponding segment bitrates.} \rev{Note that the cellular link capacities and the bitrates are measured in the same unit of Mbps.} Available bitrate set is $\{0.2, 0.4, 0.7, 1.3, 2.3\}$Mbps.
		
		With the proposed auction mechanism, as shown in Figure \ref{fig-modify:before}, two unexpected results happen due to the low capacity of user B during the first $100$ seconds: i) bitrate degradation; ii) rebuffer. For example, although user A achieves  a video bitrate of $2.3$Mbps \rev{most of the time}, a bitrate degradation  to $0.7$Mbps happens at about second 100 when user B downloads for user A. The reason is that \rev{user B has a quite low link capacity, so users A chooses a lower bitrate (when asking user B to help downloading) to avoid rebuffer.} Similar situation happens when user B downloads for himself at about   second $60$. Moreover, as user B partially relies on the downloading by himself during the first 100 seconds, he experiences rebuffer at second 50. \revk{ The rebuffer  continues until the corresponding segment has been downloaded at the end of second 53.}

		The observation in Figure \ref{fig-modify:before} motivates us to modify our proposed mechanism to \rev{increase the long-term social welfare by avoiding unexpected bitrate degradation and rebuffer}. The basic idea is that any bidder $m$ can ``skip" the available network downloading  resources from an auctioneer $n$  by refraining from  bidding if both of the following conditions are satisfied: (i) the link capacity of auctioneer $n$ is  low so that the downloading (by auctioneer $n$) for user $m$ will result in rebuffer; (ii) the link capacity of auctioneer $n$ is lower than the downloading capacity that allocates to user $m$, which is the sum of the capacities that each of user $m$'s encountered users  allocates to user $m$ (under the assumption that user $m$'s encountered user $i\in\N_m$ equally allocates his capacity to his encountered users $\N_i$). 
		Mathematically, \rev{We  introduce  coefficients $\alpha^{\textsc{link}}$ and $\alpha^{\textsc{buf}}$ to adjust bidder's willingness of refraining from bidding: a smaller coefficient indicates a smaller  willingness to skip the current resources.}

			\begin{mechanism}[\revh{Modification of Bidding Participation in  Mechanism \ref{mech:Vickrey-score}}]\label{mech:modification}
				\rev{To improve the long-term social welfare, we modify \revh{Mechanism \ref{mech:Vickrey-score}} by  allowing bidders to refrain from bidding if necessary. Specifically, }after an auctioneer $n$ announces the start of the auction with the allocation and payment rules, a bidders \rev{should refrain} from bidding if both of the following inequalities are satisfied:
				\begin{equation}\label{eq-modify:buffer}
					{h}_{n}(t)<\alpha^{\textsc{buf}}\cdot \frac{R^{\textsc{pre}}_{m,t}\cdot \beta_{m}}{{B}^{\textsc{cur}}_{m,t}}, {h}_{n}(t)< \alpha^{\textsc{link}}\cdot\sum_{i\in\mathcal{N}_m}\frac{{h}_i(t)}{|\N_{i}|},
					\end{equation}
					where $|\N_{i}|$ denotes the total number of user $i$'s encountered users. \revh{This means that only a subset of set $\M$ may choose to partipate in the bidding process. The rest of the auction is the same as Mechanism \ref{mech:Vickrey-score}.} 
			\end{mechanism}
			
         Notice that the values of the coefficients $\alpha^{\textsc{buf}}$ and $\alpha^{\textsc{link}}$ will impact on the social welfare, hence should be chosen carefully through experimental studies. Under the experiment setting similar as that in Figure \ref{fig-modify:before},  we evaluate each of the coefficient pairs $\alpha^{\textsc{link}}\in [0,2]$ and $\alpha^{\textsc{link}} \in [0,2]$ for 1000 randomly generated link capacity scenarios, and find that choosing $\alpha^{\textsc{link}} = 0.5$ and $\alpha^{\textsc{buf}} = 1$ will lead to the largest long-term social welfare on average in this experiment. Hence, in our later experiments, we set $\alpha^{\textsc{link}} = 0.5$ and $\alpha^{\textsc{buf}} = 1$. 
			
			After the modification, Figure \ref{fig-modify:after} shows the performance of the same two users (as in Figure \ref{fig-modify:before}) under the modified Mechanism \ref{mech:modification}, and we notice that the quality degradation and rebuffer do not occur (under the same experiment settings). Moreover, the modification does not have much impact on the scheduling when both users have relatively high average capacities (i.e., the last 100 seconds). Overall, the modification increases the long-term average  social welfare by $6.17\%$. 

		\begin{table}[t]  
			\center 
			\caption{Comparison between Unmodified and Modified Mechanisms.}\label{table-modify:compare}
			\begin{tabular}{c|ccccc}   	            
				\hline     
				User B's Average Capacity (Mbps) &   0.15 &   0.3 &   0.45 & 1.5  &  3.0               \\ 
				\hline
				Social Welfare Improvement ($\%$)&   $16.9$ &  $13.2$ &  $9.6$ & $0.0$ &  $0.0$ \\
				
				Rebuffer Reduction ($\%$) & $1.6$  & $0.7$ &  $0.9$ &  $0.0$ & $0.0$\\
				Bitrate Degrade Reduction ($\%$) & $22.1$  & $5.9$ &  $0.2$ &    $0.0$& $0.0$\\
				\hline
			\end{tabular}
		\end{table} 
		
			
			We further perform comparisons \rev{between the \emph{unmodified} Mechanism \ref{mech:Vickrey-score} and the \emph{modified} Mechanism \ref{mech:modification} over  1000 randomly generated network scenarios. In the experiments, user A and user B watch two different  100-second videos. The average link capacity of user A is $3$Mbps, while the average capacity of user B varies from 0.15Mbps to 3Mbps (listed in Table \ref{table-modify:compare}). The rest of the settings are the same as in Figure \ref{fig-modify:before} and \ref{fig-modify:after}. \rev{Table \ref{table-modify:compare} shows the average results over the 1000 experiment rounds. 
					As shown in the table, when user B's capacity is low (i.e., $0.15$, $0.3$, and $0.45$ Mbps)}, the modification increases the social welfare as well as reduces the rebuffer ratio (i.e., the ratio of the total rebuffer time to the total video length) and the bitrate degradation ratio (i.e., the ratio of the bitrate degradation amount to the sum of the bitrates of all the received video segments). 
				\rev{As user B's capacity becomes large (i.e., $1.5$ and $3$ Mbps), the modified and unmodfied mechanisms achieve the same performance. This is an expected result because, when both users have high capacities, the unmodified mechanism already has no rebuffer and bitrate degradation and hence no need for modification.}}

\section{Demonstration System}\label{sec:demo}
We implement \rev{the CMS system} on Raspberry PI Model B+ with the Wheezy-Raspbian operating system. 
{In the demonstration system, Raspberry PIs correspond to the mobile devices, which are equipped  with monitors (for video playing), 
	LTE USB modems (for LTE connections), and WLAN adapters (for WiFi connections). The devices can dynamically join  and leave the cooperative group and there is no need for a centralized control. After joining the cooperative group,  the mobile devices download video segments via LTE and forward messages as well as video segments to other devices (if needed)  through WiFi connections.}


\begin{figure}[t]
	\center
	\includegraphics[height=3.2cm]{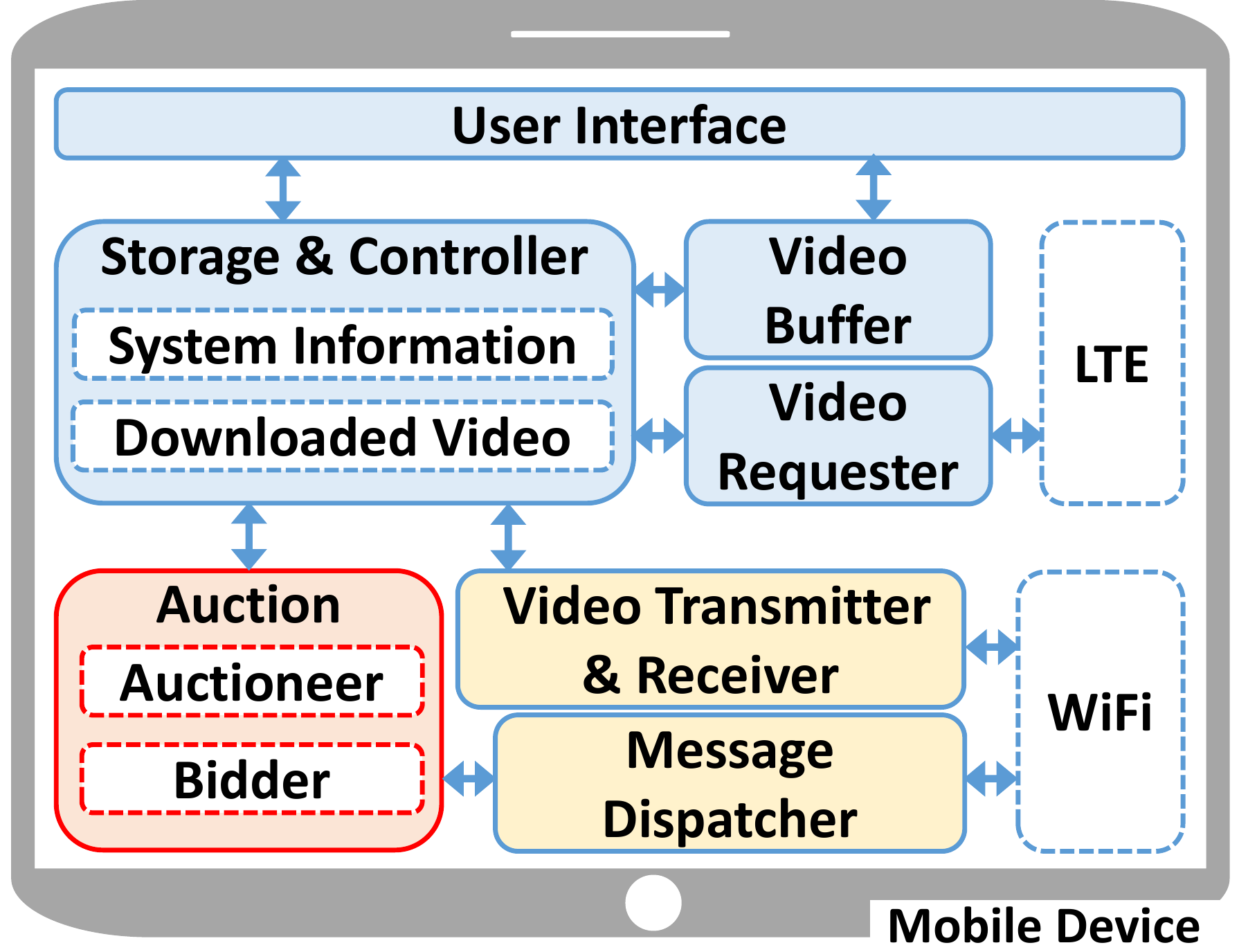}	\includegraphics[height=3.45cm]{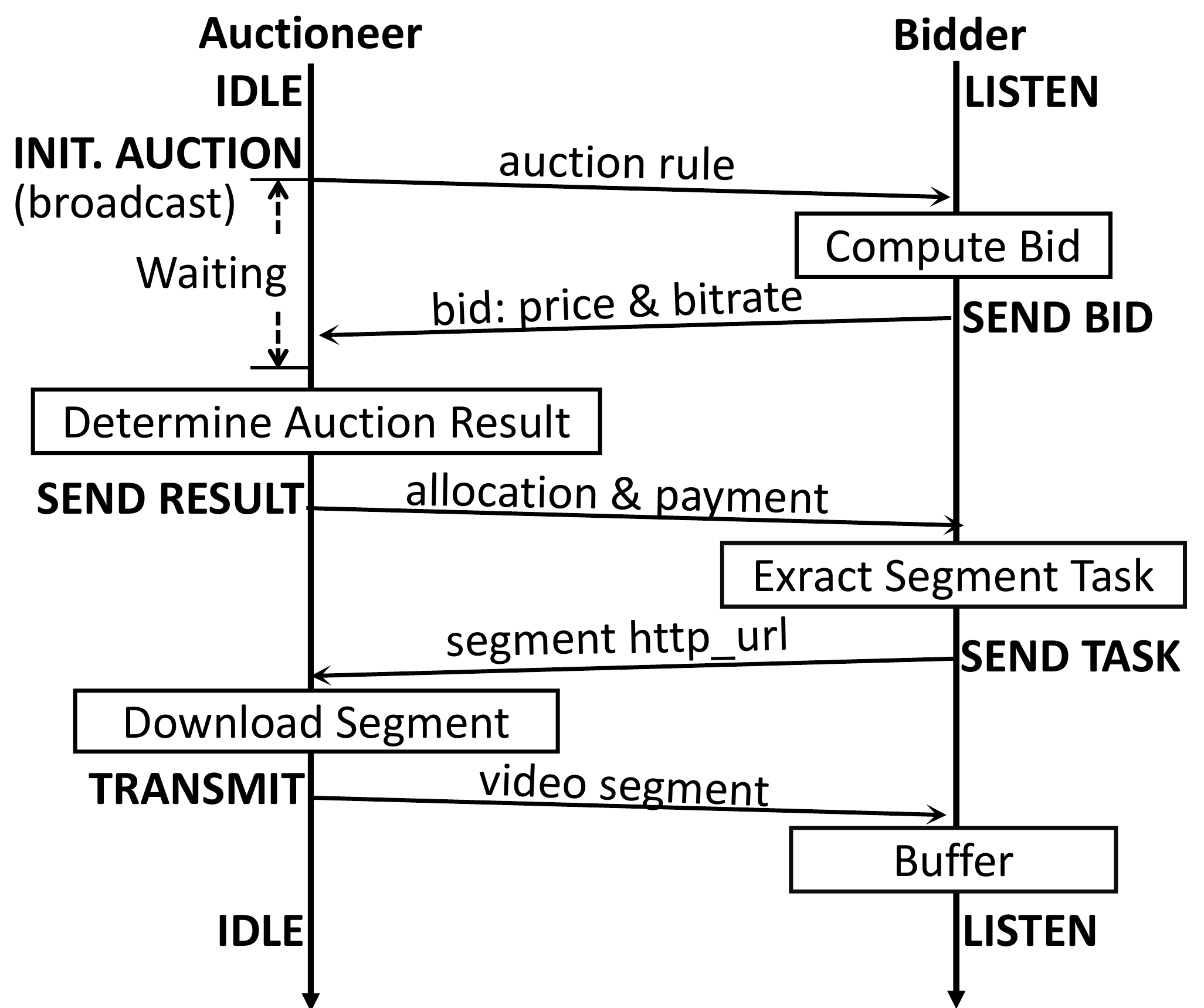}\\
	~~~(a)~~~~~~~~~~~~~~~~~~~~~~~~~~~~~~(b)
	\caption{{Demonstration System: (a) System Architecture; (b) Signaling.}}
	\label{fig:implementation}
\end{figure}
Figure  \ref{fig:implementation} (a) illustrates the system architecture with the following modules. 
\emph{User Interface} displays videos to human. \emph{Storage $\&$ Controller} stores {system information and downloaded videos}, and offers  other modules necessary control signals. \emph{Video Requester} pulls video segments from servers through LTE links, and \emph{Video Buffer} fetches and stores the segments that are for the device's own video consumption. 
{\emph{Auction}} implements our proposed auction mechanism, mainly consisting of  \emph{Auctioneer} and \emph{Bidder} modules.  When the device acts as an auctioneer, \emph{Auctioneer} module  is active and is in charge of the information announcement and auction determination. When the  device acts as a bidder, \emph{Bidder} module is active and is in charge of the bid calculation and submission. 
{\emph{Message Dispatcher}} transmits and receives auction information, such as auction  announcement and bid submission, through WiFi connections. 
{\emph{Transmitter $\&$ Receiver}} transmits the downloaded segment to others and receives the segments downloaded by others through WiFi connections. 

\rev{Figure \ref{fig:implementation} (b) shows the signaling between auctioneer's \emph{Auctioneer} module and bidders' \emph{Bidder} modules. The auctioneer first initiates the auction, then, the bidders compute and submit their bids. Since information exchange (e.g., auction initiation and bidding) takes time (due to message passing), we introduce a waiting time ($100$ms) between the auction initiation and the auction result determination to ensure that all the bids are received before determining auction results. 
After the  auction result determination, the auctioneer  announces the results to all the bidders. The winners will send the required segment URL to the auctioneer, and the auctioneer will download the segments and pass to the winners  accordingly. }

\section{Experiments and performance}\label{sec:experiment}
\begin{figure}[t]
	\center
	\includegraphics[height=2.75cm]{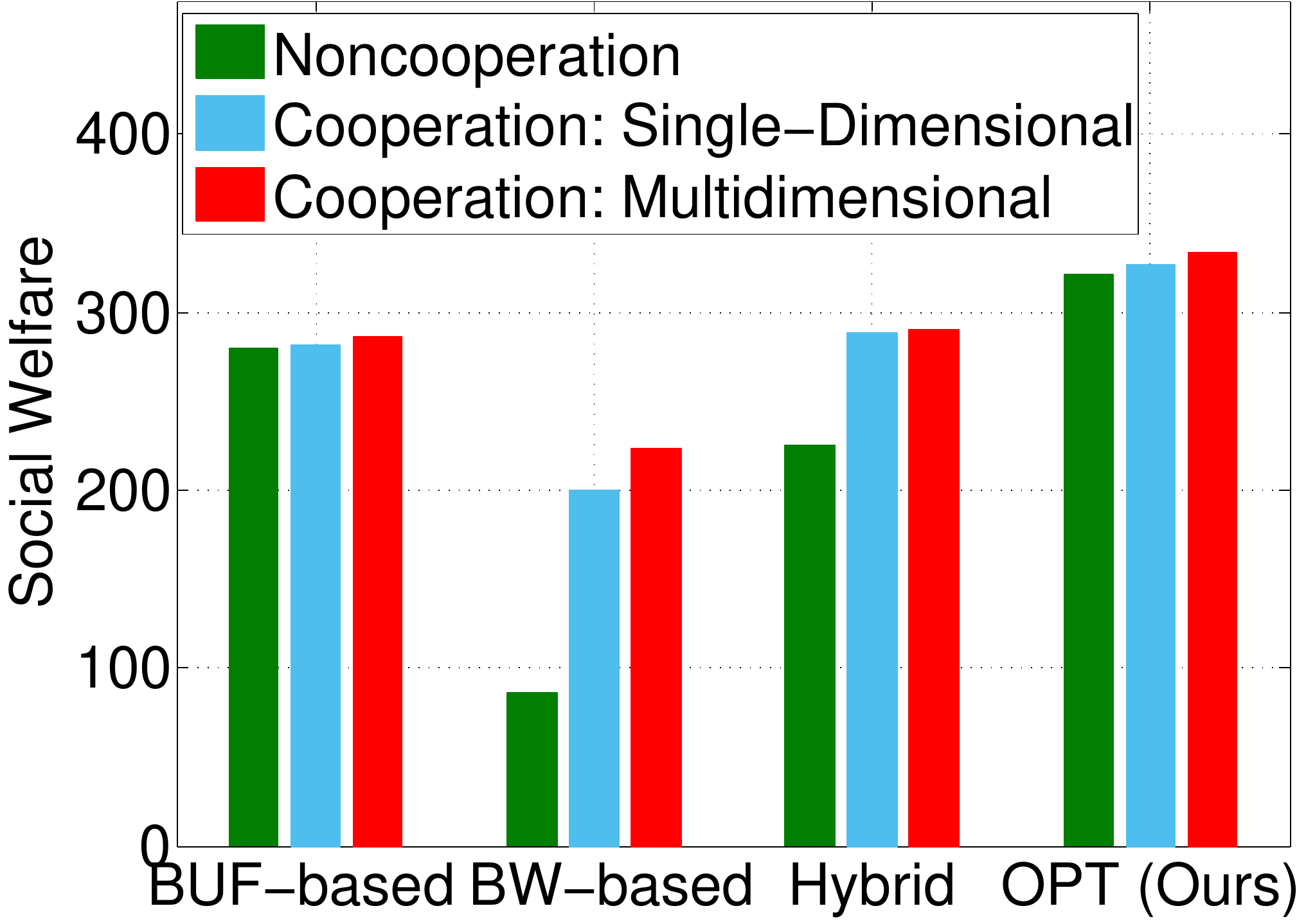}~\includegraphics[height=2.75cm]{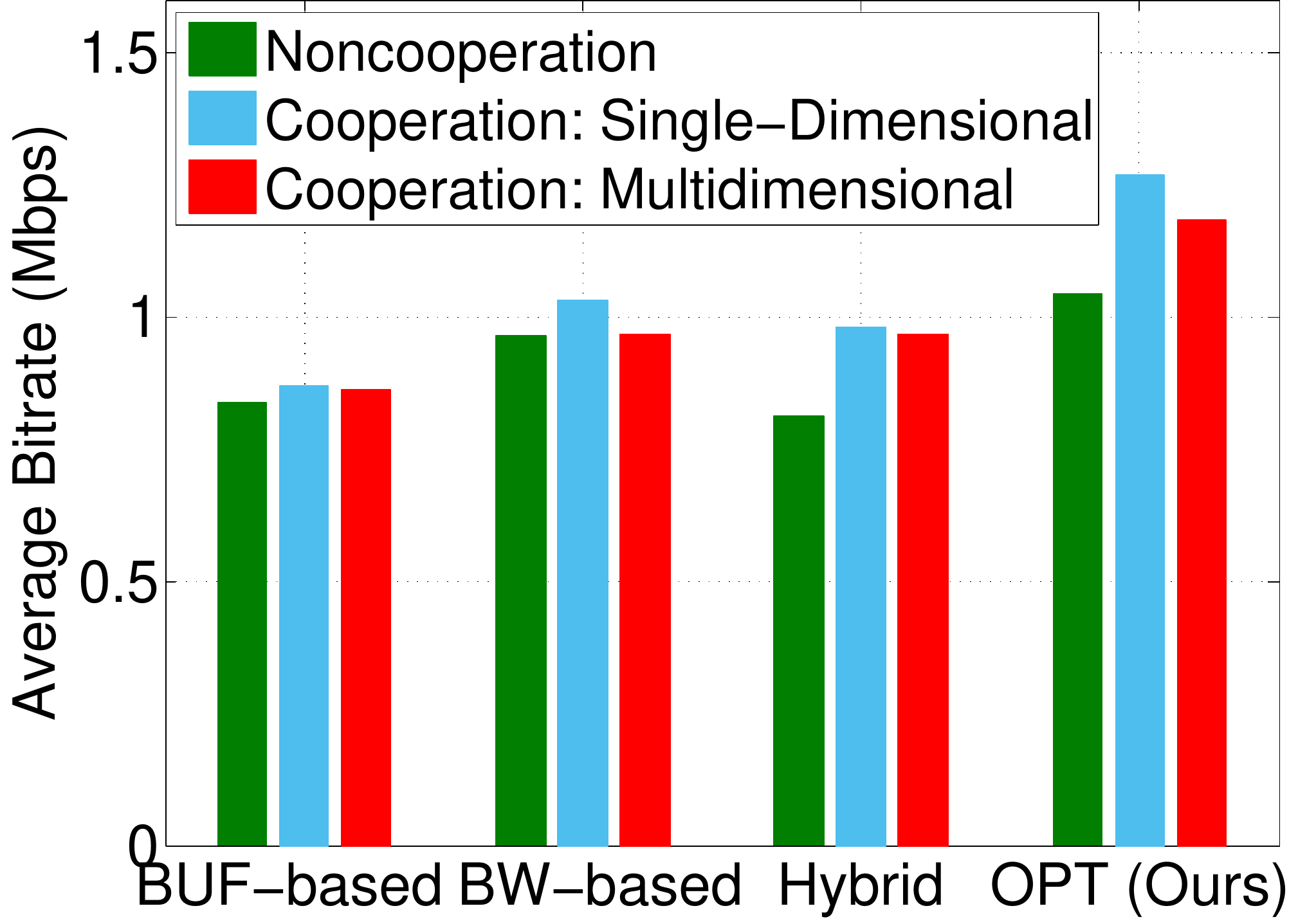}\\
	~~~~~(a)~~~~~~~~~~~~~~~~~~~~~~~~~~~~(b)\\
	\includegraphics[height=2.75cm]{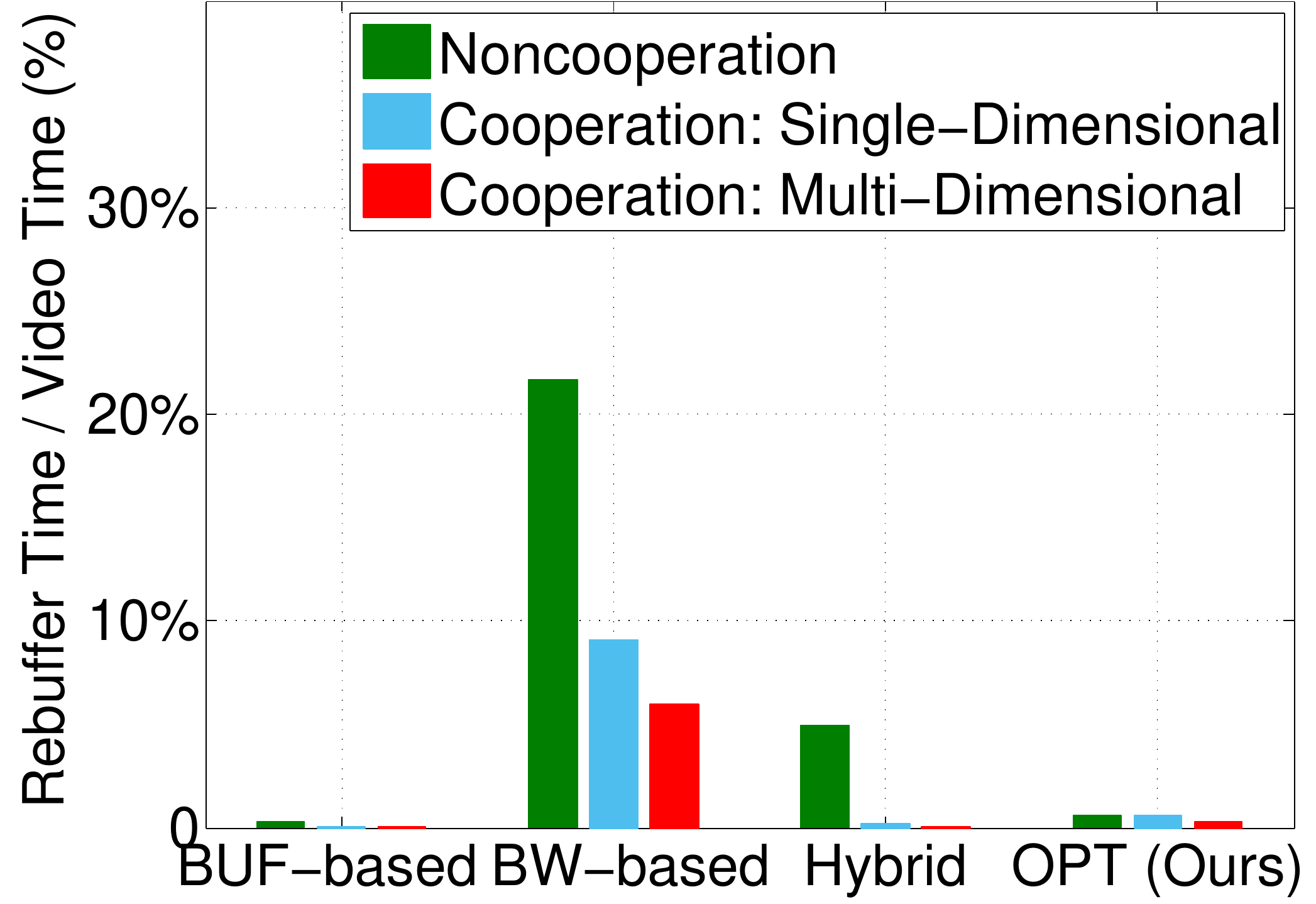}\includegraphics[height=2.8cm]{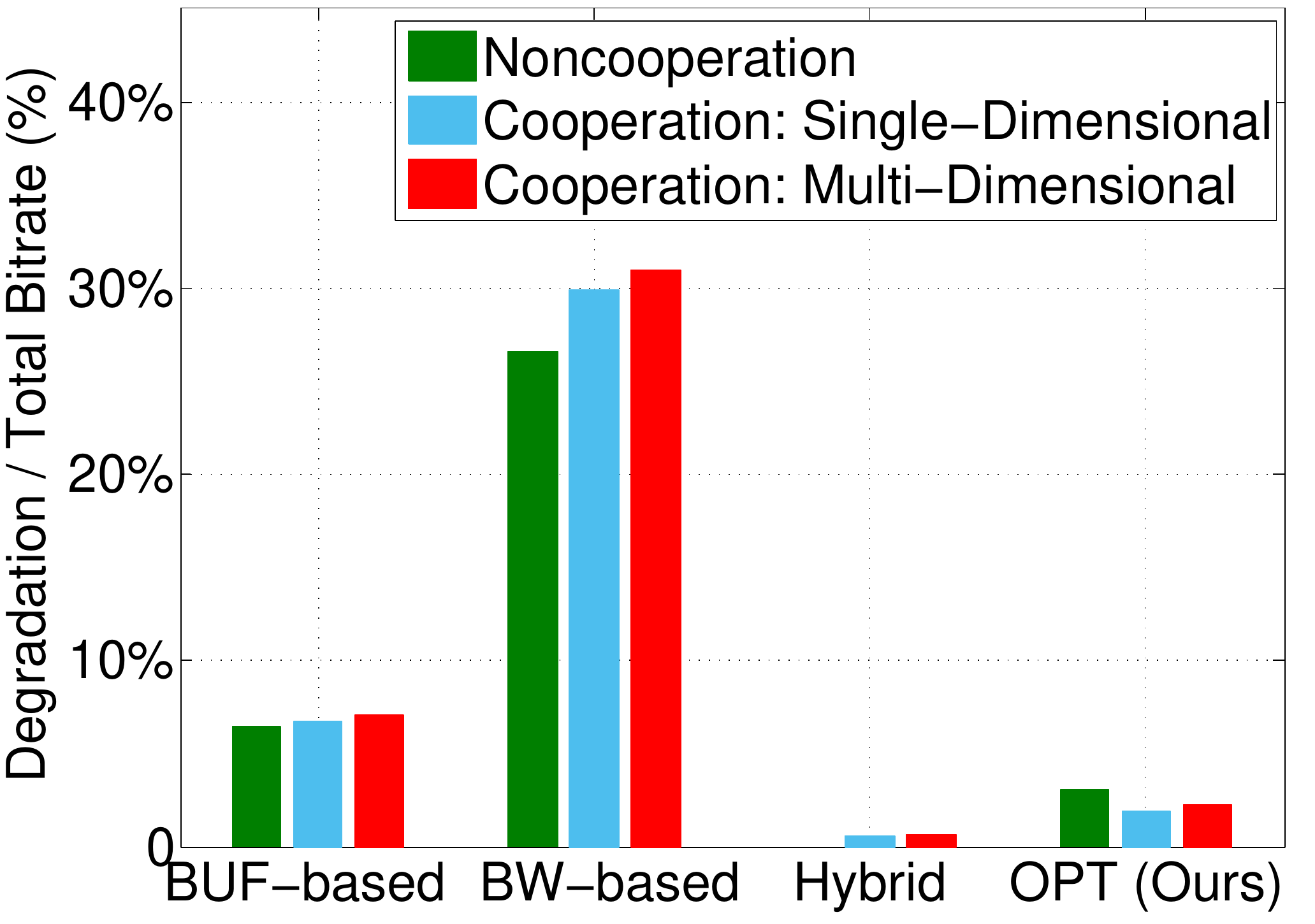}\\
	~~~~~(c)~~~~~~~~~~~~~~~~~~~~~~~~~~~~(d)\\
	\caption{Comparisons: (a) Social Welfare; (b) Average Bitrate; (c) Rebuffer; (d) Quality Degradation.}\label{fig:method}
\end{figure}
\rev{The experiments in this section are based on the modified multi-object auction mechanism (Mechanism \ref{mech:modification}). Note that the multi-object mechanism includes  the single-object mechanism as a special case by letting $K=1$.  }



\subsection{{Method Comparison}} \label{sec:multi-method}

{In this section, we compare our proposed auction scheme  with existing methods using real cellular link capacity traces obtained from BesTV.  We perform the comparison results for 500 randomly generated network scenarios  and show the average results. 
	For each network scenario, we consider 3 users whose cellular link capacities are randomly generated based on the statistics extracted from real traces, and each  user is interested in watching a 100-second video. The available bitrates for all three users' videos are \{0.2, 0.4, 0.7, 1.3, 2.3\}Mbps, and the common segment length $\beta=10$s. 
	
	We compare our mechanism with existing methods in  two aspects: (i)  comparison among \rev{three cooperative scenarios}---noncooperation, cooperation with single-dimensional (Vickrey) auction \cite{VCG1}, and cooperation with multi-dimensional (our proposed Vickrey-score) auction; (ii) bitrate adaptation comparison among buffer-based method (\emph{BUF-based})\cite{ABR-single-Huang15}, bandwidth-based method (\emph{BW-based})\cite{ABR-single-Li14}, hybrid buffer-bandwidth method (\emph{Hybrid})\cite{ABR-single-Hao14}, and our optimal bitrate method (\emph{OPT}). For now we do not consider the impact of auction overhead \rev{(i.e., auction time and energy consumption)}, hence it is  optimal to choose  $K=1$ segment due to its  maximum flexibility to the users. We will consider the impact of overhead and the proper choice of $K$ in Section \ref{subsec:overhead}. 
	
	Figure \ref{fig:method} shows the results. For comparison (i),  \rev{under each of the cooperative scenarios, we take the average among all four  methods.} Compared with noncooperation, cooperation with multi-dimensional auction increases the social welfare by $48.6\%$, {increases the average bitrate by $8.9\%$, and reduces the rebuffer by $73.7\%$}. Compared {with the cooperation with single-dimensional auction,  the cooperation with multi-dimensional auction {reduces the rebuffer by $61.4\%$ (as the multi-dimensional auction considers the bitrate adaptation)} and increases the social welfare by $3.9\%$.} For comparison (ii), {under the scenario of the cooperation with multi-dimensional auction}, our mechanism has the highest social welfare {(outperforming the other methods  by $24.8\%$ on average)}, the highest bitrate  {(outperforming the other methods by $25.8\%$ on average)}, a  relatively low rebuffer time (0.26 second on average for  a 100-second video), \revk{and a relatively low quality degradation 
	(with a degradation ratio\footnote{\revk{The quality degradation ratio is defined as the ratio of the bitrate degradation volume to the sum of the bitrates of all the received video segments. For example, for a sequence of received segments with bitrates $\{1.3,0.7,1.3\}$, the  degradation ratio is computed as $(1.3-0.7) /(1.3+0.7+1.3) = 18.2\%$.}} of $2.5\%$ on average)}.
	} 
	\subsection{{Auction Overhead}}\label{subsec:overhead}
	\begin{figure}
		\center
		\includegraphics[height=3cm]{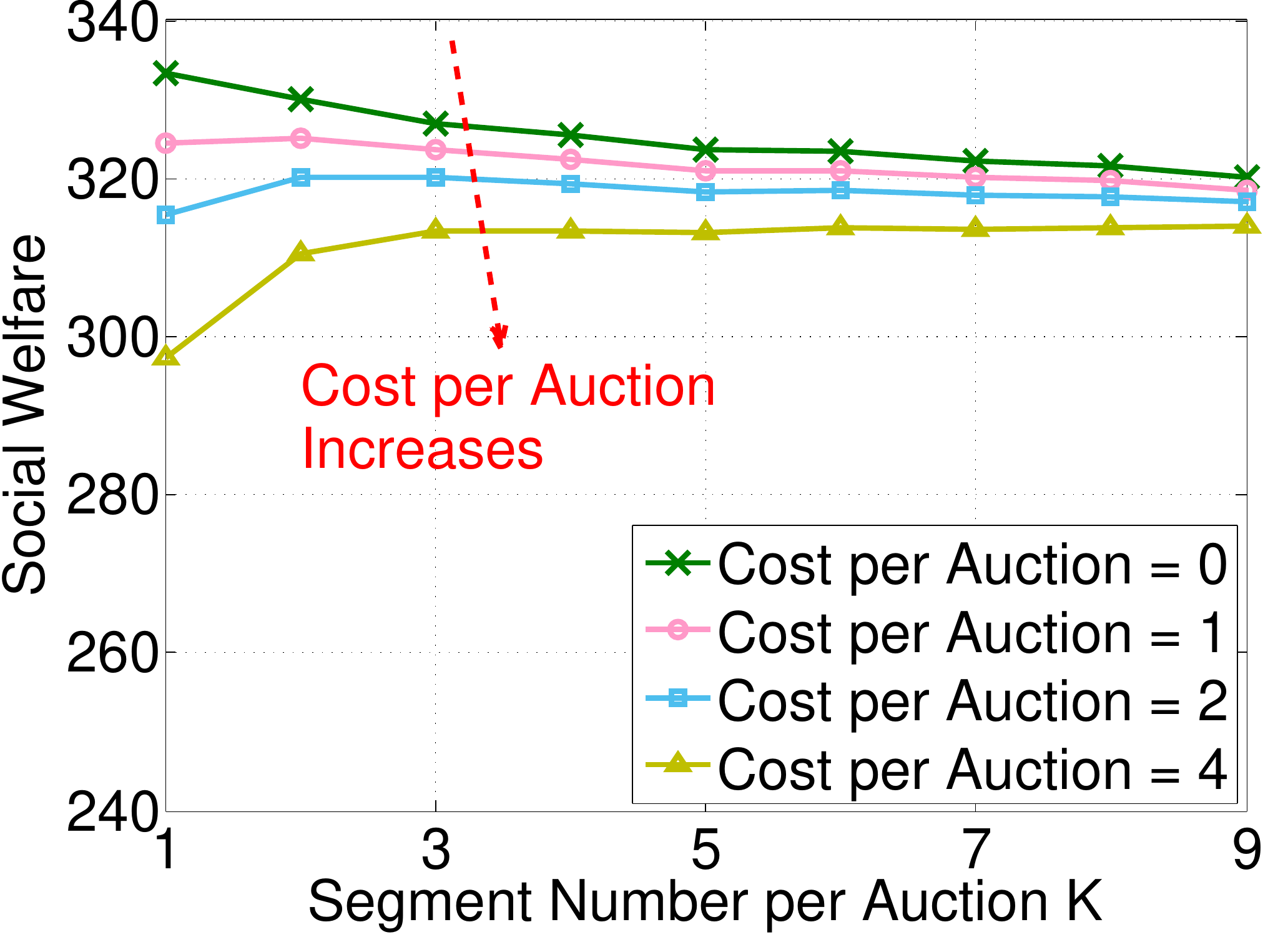}~\includegraphics[height=3cm]{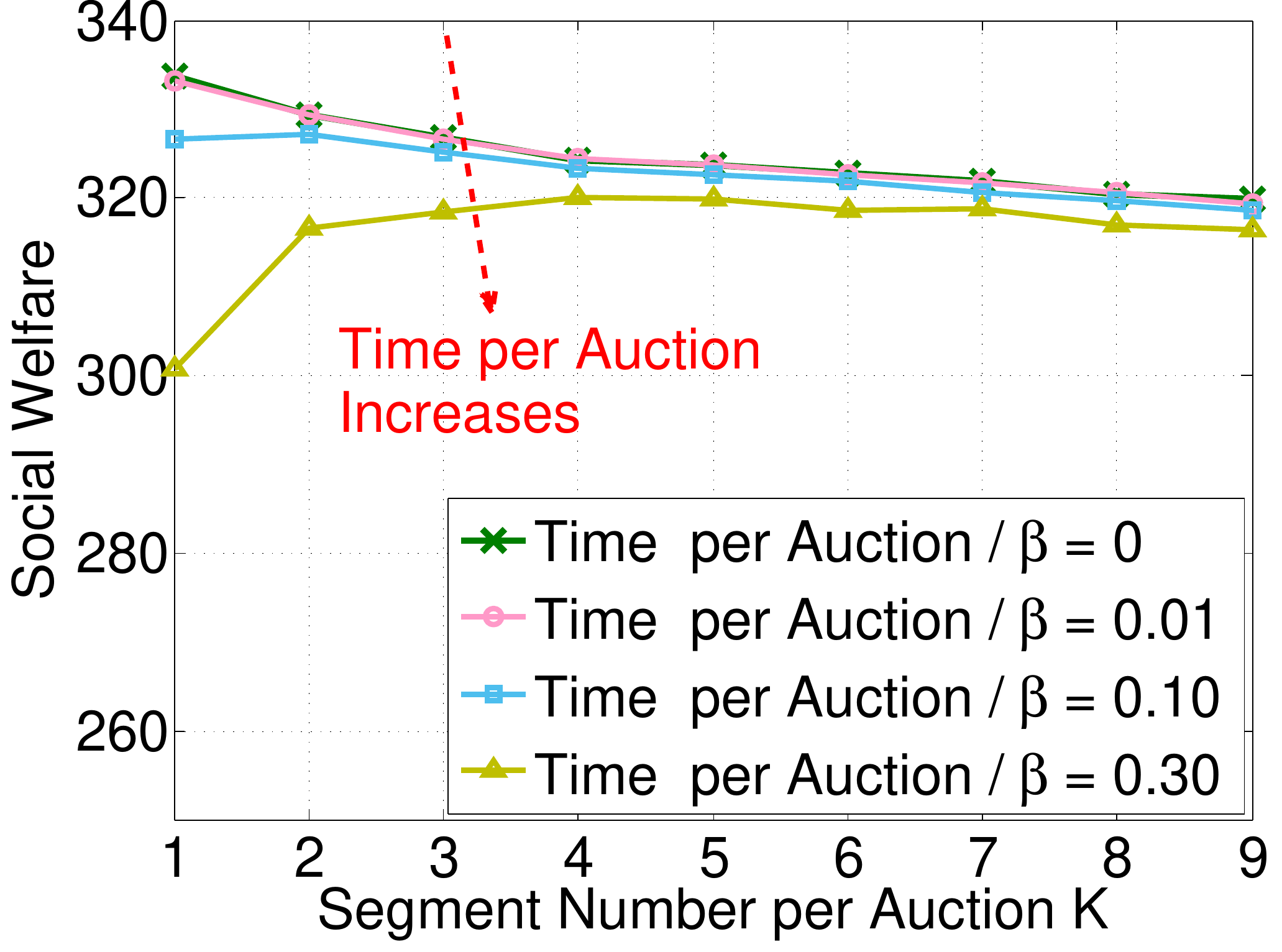}\\
		~~~~~(a)~~~~~~~~~~~~~~~~~~~~~~~~~~~~~~~(b)\\ 
		\caption{Auction Overhead: (a) Energy Consumption; (b) Time Consumption.}\label{fig:overhead1}
		
	\end{figure}
	Now we study the impact of the auction overhead and the proper choice of $K$. {Auction mechanism mainly induces two kinds of overheads: energy consumption and time consumption. By increasing the segment number $K$ per auction, {both the energy and the time spent on the auctions in a fixed  video scheduling cycle (e.g., 100 seconds in our experiment) reduce} due to less auctions. {We evaluate these two kinds of auction overheads separately.}  {The simulation setting is similar to that of Figure \ref{fig:method}, except we will change the value of $K$. } 
		
		For energy consumption, we assume that there is a fixed  \emph{cost per auction}, as in Figure \ref{fig:overhead1} (a). When the cost per auction is zero, social welfare decreases with the segment number $K$ due to the difficulty in accurately predicting future channel conditions when auctioning a larger number of segments in a single auction. As the cost per auction increases, the social welfare decreases, but a larger $K$ may be better than $K=1$  because of its smaller total overhead. {For the time consumption, we assume that there is a fixed \emph{time per auction} as in Figure \ref{fig:overhead1} (b), and we consider different ratios between this time per auction with the video segment length $\beta$. As time per auction increases, social welfare decreases, and a larger $K$ becomes better than $K=1$ because of its smaller time waste.}}

	\subsection{Realistic Performance \rev{over the Demo System}}
	We further perform experiments over the demo system introduced in Section \ref{sec:demo}. The bitrates set is \{0.5,  1.0, 2.2, 5.0\}Mbps, and the segment length $\tseg = 10s$.
	\subsubsection{Welfare Increase for High and Low Capacity Users}
	In this experiment, four users \{A,B,C,D\} form a group \rev{in a CMS system}: user A and B do not watch videos and have cellular link capacities around $3.5$Mbps; user C and D watch two different videos and have cellular capacities around $1.2$Mbps. 
	
	Figure \ref{fig-paper3:cooperation} shows the video scheduling results of users C and D in one experiment. The meanings of curves and stems are similar as that in Figures \ref{fig-modify:before}.  
	In Figure \ref{fig-paper3:cooperation}, although user C and D have link capacities  around $1.2$Mbps, they can download videos at the bitrate of 2.2Mbps most of the time and do not suffer from  rebuffer due to the help from users A and B. 
	
	\begin{figure}[t]
		\center
		\includegraphics[height=4cm]{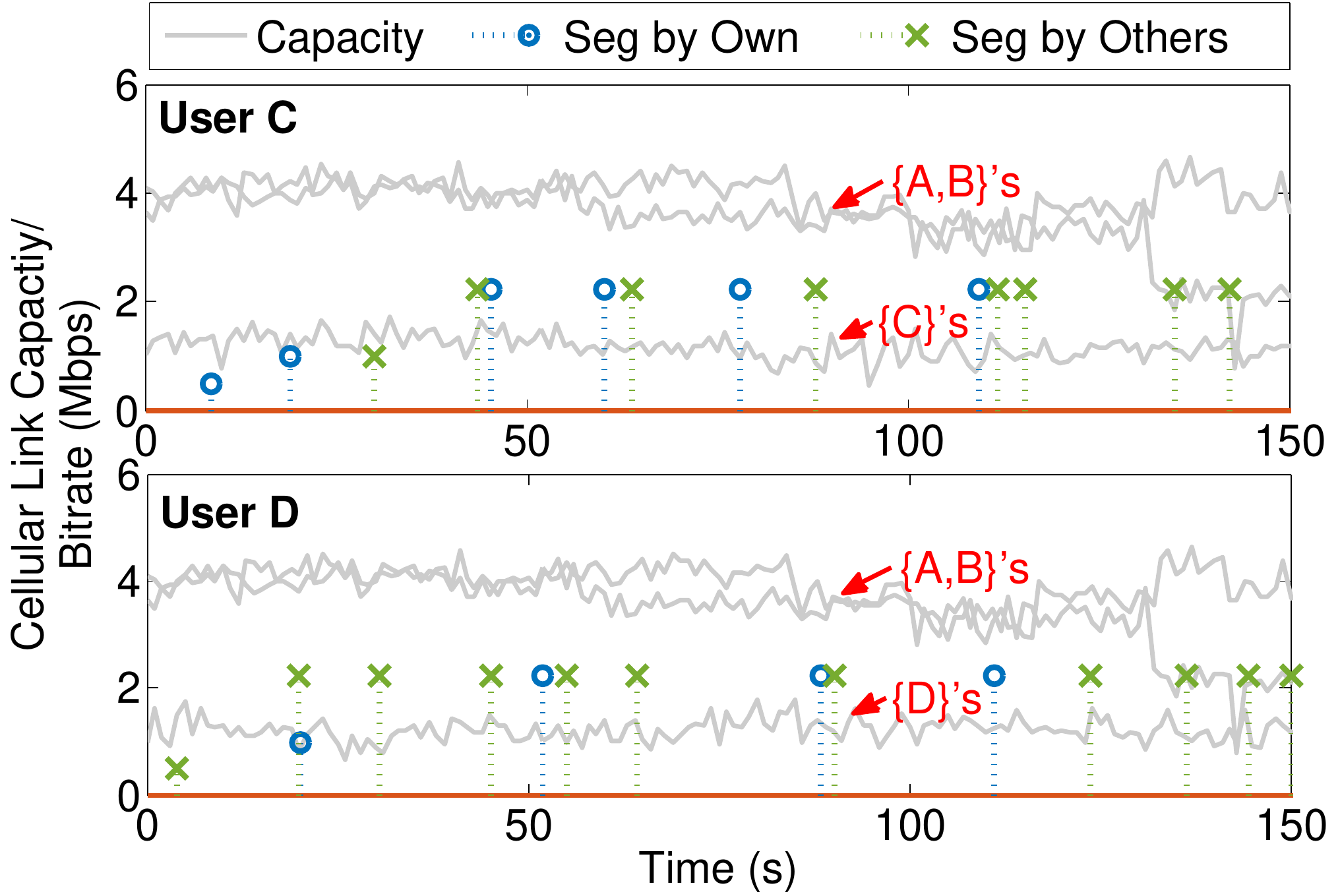}
		\caption{{Scheduling: User C and User D.}}
		\label{fig-paper3:cooperation}
	\end{figure}
	\begin{table}[t]
		\begin{center}
			\caption{Welfare Comparison} \label{table-paper3:welfare}
			\begin{tabular}{ccccc} 
				\toprule
				& 	Noncooperation & Cooperation\\
				\hline  
				A and B & $0\%$  &   $15.5\%$  \\ 
				C and D& $49.1\%$  & $84.5\%$ \\
				Social Welfare & $49.1\%$  &   $100\%$\\
				\hline
			\end{tabular}
		\end{center}
	\end{table}
	{Table \ref{table-paper3:welfare} shows users' average  normalized welfare over  four experiment rounds.   
		We normalize the social welfare (i.e., the sum of all the users' welfares) in cooperation  as $100\%$. 
		Without cooperation, users A and B receive zero social welfare, as they do not watch videos. Users C and D receive less than $50\%$ of the cooperative total social welfare. Under cooperation, users A and B receive $15.5\%$ of the social welfare due to the payments from the auction (subtracting their own costs for helping other users).  For user C and D, their welfare increases $35.4\%$ compared with noncooperation due to the service enhancement. The overall  social welfare also increases $50.9\%$ compared with noncooperation.}

	\subsubsection{Video Streaming Stability}
	\begin{figure}[t]
		\center
		\includegraphics[height=4cm]{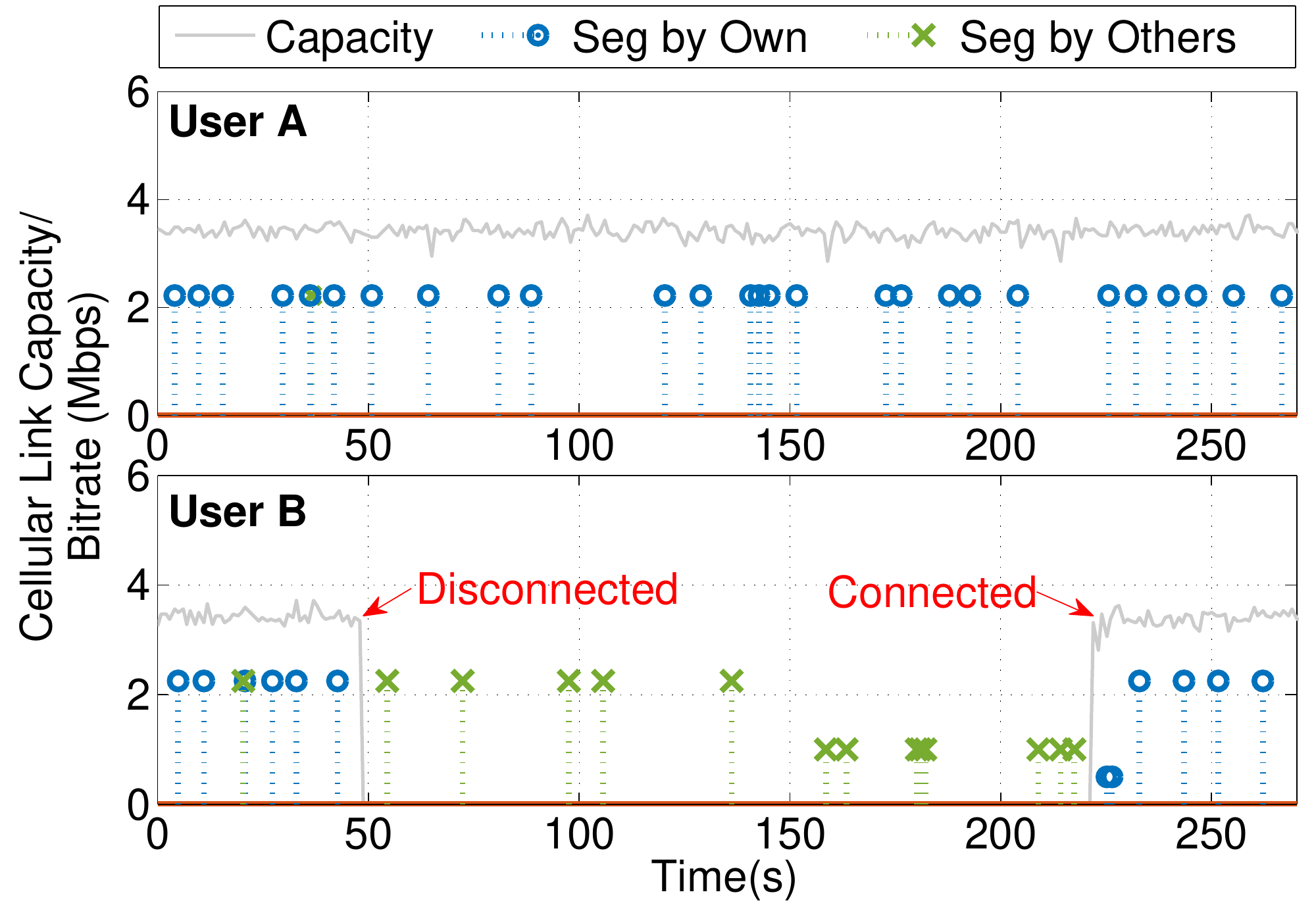}
		\caption{{Scheduling: User $B$ is Disconnected during $50\sim220$s.}}
		\label{fig-paper3:stability}
	\end{figure}
	
	We consider two users, A and B, both of which watch different videos and have cellular capacities around $3.6$Mbps. User A is always connected to the Internet, while  user B is disconnected from the Internet between 50 to 220 seconds.
	Figure \ref{fig-paper3:stability} demonstrates  the result of an experiment. The notations are similar to that of Figure \ref{fig-paper3:cooperation}. Although user B's video bitrate decreases from $2.2$Mbps to $1.0$Mbps during the time he is disconnected from the Internet, he is still able to watch the video with the help from user A. This demonstrates the practical benefit of \rev{the CMS system}.

\section{Conclusion}\label{sec:conclusion}
{\rev{The CMS system enables mobile users to share their downloading capacities for cooperative video streaming. The success of this system  requires an effective incentive mechanism that motivates user cooperations. In this work, we propose truthful and efficient mechanisms that  maximize the social welfare.} 
	We further construct a demo system to evaluate the real world performance of \rev{the CMS system}. 
\revk{For the future work, it is interesting to design mechanisms that enable the cooperation among users who \revr{will encounter in the future}, based on the prediction of their future mobility.}

\begin{IEEEbiography}
	[{\includegraphics[width=1in,height=1.25in,clip,keepaspectratio]{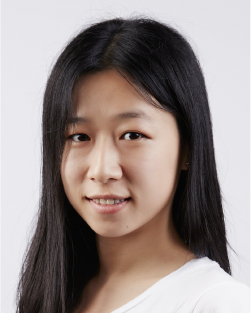}}]{Ming Tang}
	(S'16) is currently pursuing a PhD degree at the Department of Information Engineering, The Chinese University of Hong Kong (CUHK). She was a visiting student with the Department of Management Science and Engineering, Stanford University, from September 2017 to February 2018. Her research interests include wireless communications and network economics, with particular emphasis on user-provided networks and fog computing. She is a student member of the IEEE.
\end{IEEEbiography}
\vspace*{-2\baselineskip}
\begin{IEEEbiography}
	[{\includegraphics[width=1in,height=1.25in,clip,keepaspectratio]{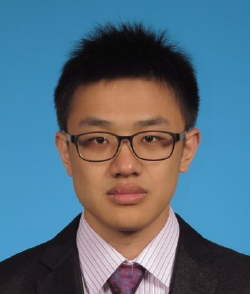}}]{Haitian Pang}
	(S'16) received his BE degree in Department of Automation in 2014 from Tsinghua University, Beijing, China. He is currently a PhD candidate in Computer Science in Tsinghua University. His research areas include edge computing edge-assisted content delivery, deep learning for network, video streaming system design, and mobile networking optimizations. He is a student member of the IEEE.
\end{IEEEbiography}
\vspace*{-2\baselineskip}
\begin{IEEEbiography}
	[{\includegraphics[width=1in,height=1.25in,clip,keepaspectratio]{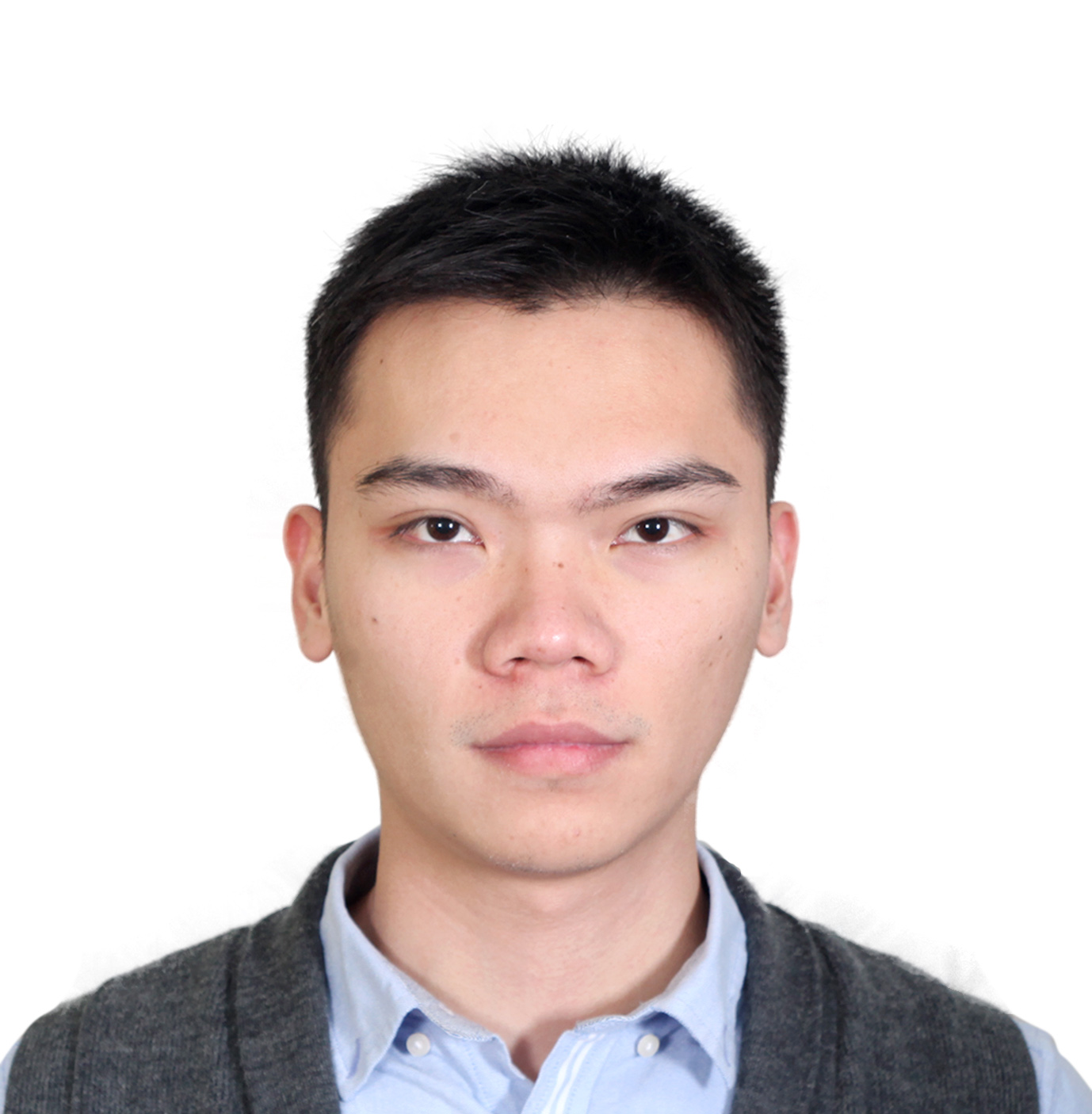}}]{Shou Wang}
	 received his BE degree and  his ME degree from Tsinghua University, Beijing, China, in 2013 and 2016, respectively. He is currently working on game development in Tencent, Shenzhen, China.  His research interests include video streaming system design, cellular-WiFi networking, and mobile bandwidth collaboration.
\end{IEEEbiography}
\vspace*{-2\baselineskip}
\begin{IEEEbiography}
	[{\includegraphics[width=1in,height=1.25in,clip,keepaspectratio]{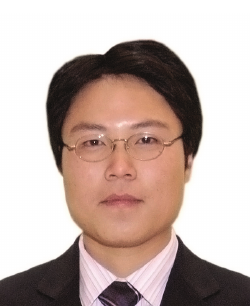}}]{Lin Gao}
	(S'08-M'10-SM'16) received the PhD degree in electronic engineering from Shanghai Jiao Tong University, in 2010. He is an associate professor with the School of Electronic and Information Engineering, Harbin Institute of Technology, Shenzhen, China. His main research	interests include network economics and games, with applications in wireless communications and networking. He received the IEEE ComSoc Asia-Pacific Outstanding Young Researcher Award in 2016. He is a senior member of the IEEE.
\end{IEEEbiography}
\vspace*{-2\baselineskip}
\begin{IEEEbiography}
	[{\includegraphics[width=1in,height=1.25in,clip,keepaspectratio]{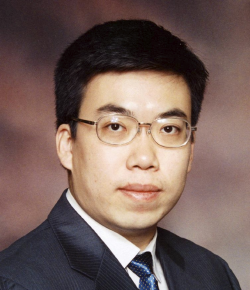}}]{Jianwei Huang}
	(F'16) is a professor with the Department of Information Engineering, The Chinese University of Hong Kong. He is the coauthor of nine Best Paper Awards, including the IEEE Marconi Prize Paper Award in Wireless Communications 2011. He has co-authored six books, including the textbook Wireless Network Pricing. He has served as the chair of the \emph{IEEE Communications Society Cognitive Networks Technical Committee} and \emph{Multimedia Communications Technical Committee}. He is an IEEE ComSoc distinguished lecturer and a Thomson Reuters Highly cited researcher. He is a fellow of the IEEE.
\end{IEEEbiography}
\vspace*{-2\baselineskip}
\begin{IEEEbiography}
	[{\includegraphics[width=1in,height=1.25in,clip,keepaspectratio]{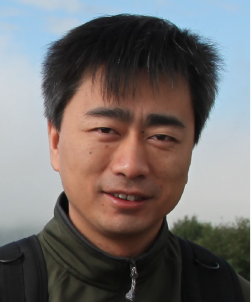}}]{Lifeng Sun} received the PhD degree in system engineer from the National University of Defense Technology, Changsha, in 2000. Currently, he is a professor with Tsinghua University. His research interests include video streaming, video coding, video analysis, and multimedia cloud computing. He received the Best Paper Award in the \emph{IEEE Transactions on Circuits and Systems for Video Technology}, in 2010, Best Paper Award at ACM Multimedia 2012, and Best Student Paper Award at MMM 2015 and IEEE BigMM 2017. He is member of the IEEE and ACM.
\end{IEEEbiography}
%

\appendix

\subsection{Proof of Proposition \ref{prop:single-truthful}}\label{app:single-truthful}
	\revk{Given any bitrate $r^m$, a bidder $m$'s score will be  $\phi^m = \Utility_{m,t}(r^m)-s(r^m)$ if bidding truthfully using price $p^m = \Utility_{m,t}(r^m)$, and $\phi' = p'-s(r^m)$ if bidding untruthfully  using price $p'\neq \Utility_{m,t}(r^m)$. We will  show that  bidder $m$ cannot obtain a higher payoff by bidding $\phi'\neq \phi^m$ (or $p' \neq p^m$), which implies that bidding with price $p^m$  (truthful  bidding) is a weakly dominant strategy for bidder $m$.
		
		According to Mechanism \ref{mech:2ndscore}, regardless of what price  that  bidder $m$ bids, he will obtain a zero payoff if he loses the auction, and he will obtain a payoff of  $P_{m}(\pr^{\dag},r^{\dag})$  if he wins,  
		\begin{equation}
		\begin{aligned}
		P_{m}(\pr^{\dag},r^{\dag}) &= \Utility_{m,t}(r^{\dag})-\pr^{\dag}  \\&= \Utility_{m,t}(r^{m})  - (\max\{\boldsymbol{\phi}_{{\N}_n/m}\} + s(r^m ))
		\\&= \phi^m- \max\{\boldsymbol{\phi}_{{\N}_n/m}\},
		\end{aligned}
		\end{equation}
		where $\max\{\boldsymbol{\phi}_{{\N}_n/m}\}$ is the maximum score other than bidder $m$'s. Hence, if bidder $m$ loses (or wins) under both  $\phi'$ and $\phi^m$, he will gain the same payoffs under both the scores. If he loses  under $\phi'$ and wins under $\phi^m$, he will gain a zero payoff under $\phi'$, and gains a payoff of $ \phi^m- \max\{\boldsymbol{\phi}_{{\N}_n/m}\} > 0 $ under $\phi^m$. If he wins under $\phi'$ and loses under $\phi^m$, he will gain a zero payoff under $\phi^m$, and gains a payoff of $ \phi^m- \max\{\boldsymbol{\phi}_{{\N}_n/m}\} < 0$ under $\phi'$, where the negative payoff is due to the fact that  bidder $m$ loses  under $\phi^m$. In each of the cases above, bidder $m$ cannot obtain a higher payoff by bidding ${\phi}' \neq \phi^m$.} 

\subsection{Proof of Proposition \ref{prop:single-bitrate}}\label{app:single-bitrate}
	The key idea is to  show that, for any bid $(\bar{r}^m, \bar{p}^m)$, there always exists a bid $(\hat{r}^m, \hat{p}^m)$, which leads to  an expected payoff of bidder $m$ that is no smaller than the bid $(\bar{r}^m,\bar{p})$ does. \rev{Such a bid $(\hat{r}^m, \hat{p}^m)$ satisfies two properties:} i) the bitrate  $\hat{r}^m$ is computed based on \eqref{eq:single-bitrate}; ii) the price $\hat{p}^m$ is chosen such that  $ \phi(\hat{r}^m, \hat{p}^m)=\phi(\bar{r}^m, \bar{p}^m)$, which means that both  
	bids $(\hat{r}^m, \hat{p}^m)$ and $(\bar{r}^m, \bar{p}^m)$ lead to the same score and hence  the same winning probability. If bidder $m$ loses the auction under such a score, then the payoff will be zero under both bids. If bidder $m$ wins the auction under such a score, then  $(\hat{r}^m, \hat{p}^m)$ leads to  a larger payoff than $(\bar{r}^m,\bar{p}^m)$, i.e.,
	\begin{multline}\label{eq:single-prpbitrate-proof}\Utility_{m,t}(\hat{r}^m) - (\max\{\boldsymbol{\phi}_{{\N}_n/m}\}+ s(\hat{r}^m))\\
	\geq \Utility_{m,t}(\bar{r}^m) - (\max\{\boldsymbol{\phi}_{{\N}_n/m}\}  + s(\bar{r}^m)).\end{multline}
	Inequality \eqref{eq:single-prpbitrate-proof} holds  because $\hat{r}^m$ satisfies equation \eqref{eq:single-bitrate}. 



\subsection{Proof of Proposition \ref{prop:multi-truthful}}\label{app:multi-truthful}

Suppose all bids \emph{except} those of  bidder $m$'s are fixed, so the $K$ highest marginal scores \emph{except} bidder $m$'s,  $\hat{\boldsymbol{S}}^{-m} = \{\hat{S}^{-m}_1,\hat{S}^{-m}_2,...,\hat{S}^{-m}_K\}$ (in the non-increasing order), are fixed. We further assume that  bidder $m$'s bitrate matrix $\brmatrix^\bidder$ is fixed. 

If bidding truthfully, bidder $m$ will submit  a price $\boldsymbol{p}^m = (p^m_1,p^m_2,...,p^m_K)$, where $p^m_{\kappa} = U_{m,t}(\boldsymbol{r}_{\kappa}^m)$ for all $\kappa$. Under such a price, bidder $m$ will win  $\kappa_m^{\dag}$ segments, and has a payoff $P_m$:
\begin{equation}P_m= U_{m,t}(\boldsymbol{r}_{\kappa_m^{\dag}}^m) - (\sum_{i=1}^{\kappa_m^{\dag}}\hat{S}^{-m}_{K-\kappa_m^{\dag}+i} + s(\boldsymbol{r}_{\kappa_m^{\dag}}^m)).\end{equation}
Let 
$\boldsymbol{S}^m = \{S^m_1,S^m_2,...,S^m_K\}$ denote bidder $m$'s marginal score vector (derived from his bids) under the truthful bidding, where $\boldsymbol{S}^m$ satisfies Assumption \ref{ass:multi-score}. Because of the truthfulness, the marginal score summation satisfies  $\sum_{i=1}^{\kappa}S_{i}^m= p^{m}_{\kappa} - s(\boldsymbol{r}_{\kappa}^m)=U_{m,t}(\boldsymbol{r}_{k}^m) - s(\boldsymbol{r}_{k}^m)$ for all $\kappa$.
Moreover, \revi{the marginal scores of those bidders who win should be no smaller than the marginal scores of those bidders who do not win. Bidder $m$ (with marginal scores  $\boldsymbol{S}^m $) wins $\kappa_m^{\dag}$ segments, so  any of the first  $\kappa_m^{\dag}$ marginal scores (winning marginal scores)  in  $\boldsymbol{S}^m $ should be no smaller than any of the last $\kappa_m^{\dag}$ marginal scores (losing marginal scores)  in $\hat{\boldsymbol{S}}^{-m} $. The bidders except bidder $m$ (with marginal scores  $\hat{\boldsymbol{S}}^{-m} $) win $K-\kappa_m^{\dag}$ segments, so any of the first $K-\kappa_m^{\dag}$ marginal scores (winning marginal scores) in $\hat{\boldsymbol{S}}^{-m} $ should be no smaller than any of the  last $K-\kappa_m^{\dag}$ marginal scores (losing marginal scores) in $\boldsymbol{S}^m $. Formally,}
\begin{equation}\label{app:prp3-score1} 
	S_{i}^m\geq \hat{S}_{j}^{-m},~i\leq {\kappa_m^{\dag}}, j\geq K-\kappa_m^{\dag}+1,\end{equation}
\begin{equation}\label{app:prp3-score2} \hat{S}_{j}^{-m}\geq S_{i}^m,~i\geq \kappa_m^{\dag}+1,j\leq K-\kappa_m^{\dag}.\end{equation}

If bidding untruthfully, bidder $m$ will submit  a price $\bar{\boldsymbol{p}}^m = (\bar{p}^m_1,\bar{p}^m_2,...,\bar{p}^m_K)$. Under such a price, bidder $m$ will win   $\bar{\kappa}_m^{\dag}$ segments, and has a payoff $\bar{P}_m $:
\begin{equation}\bar{P}_m = U_{m,t}(\boldsymbol{r}_{\bar{\kappa}_m^{\dag}}^m) - (\sum_{i=1}^{\bar{\kappa}_m^{\dag}}\hat{S}^{-m}_{K-\bar{\kappa}_m^{\dag}+i} + s(\boldsymbol{r}_{\bar{\kappa}_m^{\dag}}^m)).\end{equation}

According to above discussions, we show that bidder $m$ cannot obtain a higher payoff by submitting $\bar{\boldsymbol{p}}^m\neq \boldsymbol{p}^m$, i.e.,  we will show $P_m - \bar{P}_m
\geq0$. 
Considering  three possible situations:
\begin{itemize}
	\item If $\kappa_m^{\dag} = \bar{\kappa}_m^{\dag}$, then  $P_m-\bar{P}_m=0$. 
	
	\item If $\kappa_m^{\dag}>\bar{\kappa}_m^{\dag}$ (loses segments by untruthful bidding), then   
	\begin{equation}\label{eq:app-A-payoff2}P_m - \bar{P}_m
		=\sum_{i=\bar{\kappa}^{\dag}_m+1}^{{\kappa}^{\dag}_m}S^{m}_i - \sum_{i=1}^{\kappa_m^{\dag} - \bar{\kappa}_m^{\dag}}\hat{S}^{-m}_{K-\kappa_m^{\dag}+i} 
		\geq0.
	\end{equation}
	
	\item If $\kappa_m^{\dag}<\bar{\kappa}_m^{\dag}$ (gains segments by untruthful bidding), then  
	\begin{equation}\label{eq:app-A-payoff3}
		P_m - \bar{P}_m
		= -\sum_{i={\kappa}^{\dag}_m+1}^{\bar{\kappa}^{\dag}_m}S^{m}_i + \sum_{i=1}^{\bar{\kappa}_m^{\dag} - {\kappa}_m^{\dag}}\hat{S}^{-m}_{K-\bar{\kappa}_m^{\dag}+i} 
		\geq0.
	\end{equation}
\end{itemize}
\revh{Inequalities \eqref{eq:app-A-payoff2}  and \eqref{eq:app-A-payoff3} are obtained based on \eqref{app:prp3-score1} and \eqref{app:prp3-score2}.}

\subsection{Proof of Proposition \ref{prop:multi-bitrate}}\label{app:multi-bitrate}

For any bidder $m$, we will show that given any bid $(\bar{\boldsymbol{R}}^m,\bar{\boldsymbol{p}}^m)$, there always exists a bid $(\tilde{\boldsymbol{R}}^m, \tilde{{\boldsymbol{p}}}^m)$ that leads to an expected payoff (for bidder $m$) that is no smaller  than that achieved by  bid $(\bar{\boldsymbol{R}}^m,\bar{\boldsymbol{p}}^m)$. 
The bid $(\tilde{\boldsymbol{R}}^m, \tilde{{\boldsymbol{p}}}^m)$ 
satisfies two properties: i) bitrate $\tilde{\boldsymbol{R}}^m$ is obtained from Proposition \ref{prop:multi-bitrate}, ii) {the marginal score vector of the bid $(\tilde{\boldsymbol{R}}^m, \tilde{{\boldsymbol{p}}}^m)$ is the same as that of the bid $(\bar{\boldsymbol{R}}^m,\bar{\boldsymbol{p}}^m)$, 
	which implies that the two bids will win the same number of segments, denoted by $\kappa^{\dag}_m$}. 

If $\kappa^{\dag}_m = 0$, {bidder $m$'s payoff is zero under both the bids.} If $\kappa^{\dag}_m>0$, bidder $m$'s payoff under $(\tilde{\boldsymbol{R}}^m, \tilde{{\boldsymbol{p}}}^m)$ and  $(\bar{\boldsymbol{R}}^m,\bar{\boldsymbol{p}}^m)$ are as follows:
\begin{equation}\tilde{P}_m = U_{m,t}(\tilde{\boldsymbol{r}}^m_{\kappa^{\dag}_m}) - {(\sum_{i=1}^{\kappa^{\dag}_m} \hat{S}^{-m}_{K-\kappa^{\dag}_m+i} + s(\tilde{\brvector}_{\kappa^{\dag}_m}^\bidder))},
\end{equation} 
\begin{equation}\bar{P}_m= U_{m,t}(\bar{\boldsymbol{r}}^m_{\kappa^{\dag}_m}) - {(\sum_{i=1}^{\kappa^{\dag}_m} \hat{S}^{-m}_{K-\kappa^{\dag}_m+i} + s(\bar{\brvector}_{\kappa^{\dag}_m}^\bidder))}.
\end{equation}As bitrate  $\boldsymbol{R}^m$ is derived through maximizing $U_{m,t}(\boldsymbol{r})-s({\boldsymbol{r}})$ under the segment number constraints, i.e., 
\begin{equation}
	U_{m,t}(\tilde{\boldsymbol{r}}^m_{\kappa^{\dag}_m}) - s(\tilde{\brvector}_{\kappa^{\dag}_m}^\bidder) \geq U_{m,t}(\bar{\boldsymbol{r}}^m_{\kappa^{\dag}_m}) - s(\bar{\brvector}_{\kappa^{\dag}_m}^\bidder),
	\forall \kappa^{\dag}_m
\end{equation} 
Hence, $
\tilde{P}_m \geq \bar{P}_m $. This completes the proof of Proposition \ref{prop:multi-bitrate}.

\subsection{Proof of Lemma \ref{lem:bitrate-p1}}\label{app:identical}


First, we prove  that the non-zero \rev{elements} in the optimal bitrate $\boldsymbol{r}_{\kappa}^m$ should  be in the  ascending order, \rev{i.e., $r_{{\kappa}1}^m\leq r_{{\kappa}2}^m\leq...\leq r_{{\kappa}{\kappa}}^m$}.  
We prove this through contradiction. \rev{Suppose the optimal bitrate $\boldsymbol{r}_{\kappa}^m$ is not in the ascending order. By reordering the elements in $\boldsymbol{r}_{\kappa}^m$ in ascending order, we obtain  a new vector $\bar{\boldsymbol{r}}_{\kappa}$.} 
Note that $\Valueq_{m,t}(\cdot)$ and $\Valueb_{m,t}(\cdot)$ in the bidder's utility $U_{m,t}(\cdot)$ are independent of bitrate order, and the auctioneer $n$'s downloading cost  $C_{n,t}(\cdot)$ is also independent of the bitrate order. The  degradation loss $\Lossdeg_{m,t}(\cdot)$, however,  is minimized when \rev{the non-zero elements in} $\boldsymbol{r}_{\kappa}$ are in the ascending order.  \rev{Hence, we have the following inequality 
 \begin{equation}U_{m,t}(\bar{\boldsymbol{r}}_{\kappa}) - C_{n,t}(\bar{\boldsymbol{r}}_{\kappa})\geq  U_{m,t}(\boldsymbol{r}_{\kappa}^m) - C_{n,t}(\boldsymbol{r}_{\kappa}^m),\end{equation} which contradicts the definition of the optimal bitrate, i.e., $\boldsymbol{r}_{\kappa}^m=\arg\max_{\boldsymbol{r}_{\kappa}} \left(U_{m,t}(\boldsymbol{r}_{\kappa}) - C_{n,t}(\boldsymbol{r}_{\kappa})\right)$. This proves the ascending order of non-zero bitrate elements.  }

Next, we prove that the non-zero  elements in optimal bitrate $\boldsymbol{r}_{\kappa}^m$ \rev{should be} identical, i.e., $r_{{\kappa}i}^m=r_{\kappa}^m~\forall i\leq \kappa$.  
\revh{To simplify the presentation, we define a function} \rev{$g_{mn,t}(r) = \valueq_{m,t}(r) - c_{n,t}(r), r\in\R_m$. 
	Among the finite bitrate set $\R_m$, there exists an optimal bitrate, denoted by $r^{*}\in\R_m$, that maximizes the concave function $g_{mn,t}(r)$}. 
As we have shown that the non-zero bitrates will be in the  ascending order, the only possible bitrate degradation is the degradation from $\prebr_{m,t}$ \rev{(i.e., the bitrate of the last segment from the previous auction)} to $r_{{\kappa}1}^m$ (the bitrate of the first allocated segment in this auction). Hence, we can rewrite the bitrate vector optimization  problem as follows:
\begin{equation}\label{eq:cd-1}\boldsymbol{r}_{\kappa}^m =\arg \max_{\boldsymbol{r}_{\kappa}}\left(\sum_{i=1}^{\kappa} g_{mn,t}(r_{{\kappa}i})  -\lossdeg(\prebr_{m,t},r_{{\kappa}1})\right).\end{equation}

We show $r_{{\kappa}i}^m=r_{\kappa}^m,~\forall i\leq \kappa$ in the following two cases:
\begin{itemize}
	\item If $\prebr_{m,t}<r^{*}$, then we know that  bitrate $r^*$ maximizes $ g_{mn,t}(r)$ and minimizes $\lossdeg(\prebr_{m,t},r)$ among the feasible  set  \rev{$r\in\R_m$}. Hence, the optimal bitrate vector $\boldsymbol{r}_{\kappa}^m$ satisfies 
	$r_{{\kappa}1}^m = r_{{\kappa}2}^m = ... = r_{{\kappa}{\kappa}}^m = r^*.$
	\item If $\prebr_{m,t}\geq r^*$, \revh{then we prove the identical bitrate result as follows. First, we have bitrates $ r^*\leq r_{{\kappa}1}^m$ for the following reasons}. \revi{If $ r_{{\kappa}1}^m$ is the optimal value for $r_{\kappa 1}$, then the partial derivative of the objective function in \eqref{eq:cd-1} with the respect to $r_{{\kappa}1}$ at $r_{{\kappa}1}=r_{{\kappa}1}^m$ should be zero, i.e., $[g_{mn,t}(r_{\kappa 1}^m)]_{r_{\kappa 1}}-[\lossdeg(\prebr_{m,t},r_{{\kappa}1}^m)]_{r_{\kappa 1}}=0$. Hence, $[g_{mn,t}(r_{\kappa 1}^m)]_{r_{\kappa 1}}\leq 0$, which implies that $ r^*\leq r_{{\kappa}1}^m$ holds because $g_{mn,t}(\cdot)$ is concave and is maximized at $r^*$.} 
	Second, \rev{
		the concave function $g_{mn,t}(r)$ is non-increasing in  $r$ for $r\geq r_{{\kappa}1}^m \geq r^*$.}  
	\revh{We have proved   that $r_{{\kappa}1}^m\leq r_{{\kappa}2}^m\leq... \leq r_{{\kappa}{\kappa}}^m$, so \revi{$g_{mn,t}(r_{{\kappa}2}^m),g_{mn,t}(r_{{\kappa}3}^m),...,g_{mn,t}(r_{{\kappa}\kappa}^m)$ are maximized when bitrates $r_{\kappa2},r_{\kappa3},...,r_{\kappa\kappa}$ are minimized under the constraint that $r_{{\kappa}1}^m\leq r_{{\kappa}2}^m\leq... \leq r_{{\kappa}{\kappa}}^m$}, which implies $r_{\kappa 2}^m=... = r_{{\kappa}{\kappa}}^m=r_{{\kappa}1}^m  $.}  \endproof
\end{itemize}

\subsection{Proof of Lemma \ref{lem:bitrate-p2}}\label{app:non-increasing}

According to Lemma \ref{lem:bitrate-p1}, the optimal common non-zero bitrate $r_{{\kappa}}^m$ for each row $\kappa$ is derived as follows: \begin{equation}\label{eq:cd-2}r_{{\kappa}}^m=\arg\max_r\left({\kappa}\cdot g_{mn,t}(r) -\lossdeg(\prebr_{m,t},r)\right).\end{equation}
We prove that $r_{\kappa}$ is non-increasing in row index $\kappa$  by checking two cases:
\begin{itemize}
	\item  If $\prebr_{m,t}<r^*$, then bitrate $r^*$ maximizes ${\kappa}\cdot g_{mn,t}(r) -\lossdeg(\prebr_{m,t},r)$ for any  ${\kappa}$. Hence, $r_1^m = r_2^m=...=r_K^m = r^*$.
	\item If $\prebr_{m,t}\geq r^*$, then bitrate $r^*\leq r_{\kappa}^m=r_{\kappa1}^m$ \rev{(proved in Property 1)}. 
	\rev{The optimal bitrate ${r_{\kappa}^m}$  and ${r_{\kappa+1}^m}$ satisfy: \begin{equation}{r_{\kappa}^m}=\arg\max \left({\kappa}\cdot g_{mn,t}(r) -\lossdeg(\prebr_{m,t},r)\right);\end{equation} \begin{equation}{r_{\kappa+1}^m}=\arg\max \left( {\kappa}\cdot g_{mn,t}(r)-\lossdeg(\prebr_{m,t},r)+g_{mn,t}(r)\right) .\end{equation}
		Function $ g_{mn,t}(\cdot)$ is concave and non-increasing in  $r$ when $r\in\R_m$ and $r\geq r^*$. Suppose $r_{{\kappa}}^m< r_{{\kappa+1}}^m$. Then  based on the definition of $r_{\kappa}^m$, we have 
		\begin{multline} {\kappa}\cdot g_{mn,t}(r_{{\kappa+1}}^m)-\lossdeg(\prebr_{m,t},r_{{\kappa+1}}^m)+g_{mn,t}(r_{{\kappa+1}}^m)\\
			\revh{<}~{\kappa}\cdot g_{mn,t}(r_{{\kappa}}^m)-\lossdeg(\prebr_{m,t},r_{{\kappa}}^m)+g_{mn,t}(r_{{\kappa}}^m),
		\end{multline}  
		which contridicts to the definition of the optimal bitrate $r_{{\kappa+1}}^m$. Hence, $r_{{\kappa}}^m\geq r_{{\kappa+1}}^m$ for all $\kappa = 1,2,...,K-1$.}
\end{itemize}
This completes the proof for Lemma \ref{lem:bitrate-p2}. 

\subsection{Proof of Proposition \ref{cnd:non-negative}}\label{app:condition}
Considering efficient score function in \eqref{eq:efficient}, if a  bidder $m$ bids according to  Proposition \ref{prop:multi-truthful} and \ref{prop:multi-bitrate}, then the bidder's score  for being allocated a total  of $\kappa$ segments is given:
\begin{equation} \label{eq:phi}
	\begin{aligned}
		\score^{m,n,t}_{\kappa} = & \underset{\brvector}{\text{maximize}}
		& & \Utility_{m,t}(\brvector) - \Cost_{n,t}(\boldsymbol{\br}) \\
		& \text{subject to}
		& & \br_{i}>0,  \quad i=1,...,\kappa,\\
		&&&\br_{i}=0, \quad i=\kappa+1,...,K,\\
		& \text{variables}
		& &  {\br_i \in \R_m, \quad i=1,...,\kappa.}
	\end{aligned}
\end{equation}
\revh{Hence, the conditions on the marginal scores in Assumption \ref{ass:multi-score}  can be written as equivalent conditions of $\score^{m,n,t}_{\kappa}$ as follows: }	
\begin{itemize}
	\item Non-negative marginal score:
	\begin{equation}\label{eq:cnd-non-negative}
		\score^{m,n,t}_{\kappa+1} - \score^{m,n,t}_{\kappa} \geq 0, \forall \kappa=1,2,..., K-1
	\end{equation}
	\item Non-increasing marginal score:
	\begin{equation}\label{eq:cnd-non-increasing}
		\score^{m,n,t}_{\kappa} - \score^{m,n,t}_{\kappa-1} \geq \score^{m,n,t}_{\kappa+1}-\score^{m,n,t}_{\kappa}, ~\forall \kappa=1,2,...,K-1
		\end{equation}
\end{itemize}

\revh{Next, we show that inequalities \eqref{eq:non-negative} and \eqref{eq:non-increasing} are the sufficient conditions for satisfying \eqref{eq:cnd-non-negative} and \eqref{eq:cnd-non-increasing}, respectively.}

\emph{\textbf{Non-negative}}: If $\valueq_{m}(r,\theta)\geq c_{n,t}(r)$ for all $r$ and $\theta$, then $g_{mn,t}(r) = \valueq_{m,t}(r) - c_{n,t}(r) \geq 0$ always holds. Based on Lemma \ref{lem:bitrate-p1}, the score $\score^{m,n,t}_{\kappa} $ can be represented as follows: \begin{equation}\score^{m,n,t}_{\kappa} = \max_r \left( {\kappa}\cdot g_{mn,t}(r) -\lossdeg(\prebr_{m,t},r)\right).
\end{equation} \rev{Let $r_{{\kappa}}^m$ and $r_{{\kappa+1}}^m$ be the optimal non-zero common bitrates for rows $\kappa$ and $\kappa+1$, respectively.} Based on Lemma \ref{lem:bitrate-p1} and \ref{lem:bitrate-p2}, \rev{\begin{multline}
	\score^{m,n,t}_{\kappa+1} =  	{(\kappa+1)}\cdot g_{mn,t}(r_{{\kappa+1}}^m)-\lossdeg(\prebr_{m,t},r_{{\kappa+1}}^m)\\
	\geq{(\kappa+1)}\cdot g_{mn,t}(r_{{\kappa}}^m)-\lossdeg(\prebr_{m,t},r_{{\kappa}}^m)\\
	\geq\kappa \cdot g_{mn,t}(r_{{\kappa}}^m)-\lossdeg(\prebr_{m,t},r_{{\kappa}}^m)
	=\score^{m,n,t}_{\kappa},
\end{multline}
which proves that $\score^{m,n,t}_{\kappa+1} - \score^{m,n,t}_{\kappa} \geq 0$.} 


\emph{\textbf{Non-increasing}}:	The non-increasing marginal score requirement is equivalent to the following one:
\begin{multline}\label{eq:conditon3}
	{(\kappa+1)} g_{mn,t}(r_{{\kappa}+1}^m)  - 2{\kappa} g_{mn,t}(r_{{\kappa}}^m) + ({\kappa}-1) g_{mn,t}(r_{{\kappa}-1}^m)  \\ -\lossdeg(\prebr_{m,t},r_{\kappa+1}^m)  -2\lossdeg(\prebr_{m,t},r_{{\kappa}}^m) -\lossdeg(\prebr_{m,t},r_{{\kappa}-1}^m)\leq |\tilde{\Delta}|.
\end{multline}
Based on Lemma \ref{lem:bitrate-p1} and \ref{lem:bitrate-p2}, we derive the conditions for satisfying inequality \eqref{eq:conditon3} in the following two cases:
\begin{itemize}
	\item If $\prebr_{m,t}<r^*$, the bitrate $r_{{\kappa}}^m =r_{{\kappa}-1}^m=r_{{\kappa}-2}^m = r^*$. Hence, inequality \eqref{eq:conditon3} is directly satisfied.
	\item If $\prebr_{m,t}\geq r^*$, then the inequality \eqref{eq:conditon3} is satisfied if:
	\begin{equation}2{(\kappa+1)}(c_{n,t}(r_{{\kappa}}^m) - c_{n,t}(r_{{\kappa}-1}^m)) +\lossdeg(\prebr_{m,t},r_{{\kappa}}^m)\leq |\tilde{\Delta}|,\end{equation} \revh{because $g_{mn,t}(\cdot)$ is concave and non-increasing in $r$ when  $r\geq r^*$, and $\lossdeg(\cdot)$ is non-increasing in $r$.}
	\revh{Since $c_{n,t}(r_{{\kappa}}^m) - c_{n,t}(r_{{\kappa}-1}^m)\leq c_{n,t}( R_m^Z) -c_{n,t}(0) =  c_{n,t}( R_m^Z)$, $\lossdeg(\prebr_{m,t},r_{{\kappa}}^m)\leq \lossdeg_m( R_m^Z,0) $, and $\kappa+1\leq K$,} 
	we have the sufficient condition for non-increasing marginal score:
	\begin{equation}
		2K\cdot c_{n,t}( R_m^Z) + \lossdeg_m( R_m^Z,0) \leq |\tilde{\Delta}|.
	\end{equation}
\end{itemize}
This completes the proof.

\subsection{Implementation Issues}
\revk{Here, we discuss various practical implementation issues of the proposed crowdsourced framework.
	
	
	\revk{First, we consider the implementation overhead of the proposed auction, which mainly consists of the computation overhead and the transmission overhead. 
		The computing  involved in the auction  (i.e., the determination of  allocation, price, and payment) mainly invokes  a sorting algorithm, \revr{which has a low complexity (e.g.,  Quicksort algorithm \cite{quicksort} has an average time complexity of $O(NK\log(NK))$, where $N$ is the number of bidders and $K$ is the number of segments per auction).}
		The transmission overhead  mainly involves  the waiting time needed to collect all bidders' bids, which is set to be $100$ms in the demonstration system. Hence, it is much smaller than the \revr{segment downloading time. For example, if a segment has a bitrate of 3Mpbs and a playback time of 10 seconds, then the downloading time over a link with capacity of 4Mbps will be $3\times10/4 = 7.5$ seconds}.} 
	
	\revk{Next we consider the payment management in practice. One approach is to consider a credit-based system\cite{payment1}, where the users who receive helps pay credit (either virtual or true currency) to the users who help them. 
		Such a credit-based system has been applied in commercial user-provided networks, e.g., FON (https://network.fon.com). To avoid potential cheating behaviors in the payments, one needs to consider either centralized management systems \cite{payment2} or a third-party entity \cite{payment3} to supervise the transactions and service provisions.}
	
	
	Finally, we consider the  security issues (e.g., impersonation attack and denial-of-services) and the privacy issues (e.g.,  user contents protection). In this regard, we can apply the related solutions in the D2D literature that address similar issues, such as the  key management method \cite{secure} for the security issues and the video encryption method \cite{secure2} for the privacy issues.}


\end{document}